\documentclass{article}

\RequirePackage[OT1]{fontenc} \RequirePackage{amsthm,amsmath} \RequirePackage[authoryear]{natbib}

\usepackage{amsfonts, amssymb, graphicx, mathtools, epstopdf, amsmath} 
\usepackage{accents, appendix, xcolor}
\usepackage{enumitem, datetime}

\usepackage{setspace}
\usepackage{caption}
\usepackage{subcaption, multirow}
\usepackage[colorlinks=true,linkcolor=blue, allcolors=blue]{hyperref}

\usepackage{multibib}
\newcites{appendix}{References}

\newtheorem{theorem}{Theorem}
\newtheorem{assumption}{Assumption}
\newtheorem{proposition}{Proposition}

\newtheorem{lemma}{Lemma}

\DeclareMathOperator*{\argmin}{arg\,min}
\newcommand{\tmatch} {{\text{matched}}}
\newcommand{\tweight} {{\text{weighted}}}
\newcommand{\tadj} {{\text{adj}}}

\newcommand{\E}{{\text{E}}}

 \newcommand{\eff}{{\textit{Effect}}}

\newcommand{\blind}{0}

 \title{\normalsize New Estimands for Experiments with Strong Interference}
\author{\normalsize David Choi}
\date{\normalsize \today, \ \currenttime}

\begin{document}
\maketitle

\abstract{In experiments that study social phenomena, such as peer influence or herd immunity, the treatment of one unit may influence the outcomes of others. Such ``interference between units'' violates traditional approaches for causal inference, so that additional assumptions are often imposed to model or limit the underlying social mechanism. For binary outcomes, we propose new estimands that can be estimated without such assumptions, allowing for interval estimates assuming only the randomization of treatment. However, the causal implications of these estimands are more limited than those attainable under stronger assumptions, showing only that the treatment effects under the observed assignment varied systematically as a function of each unit's direct and indirect exposure, while also lower bounding the number of units affected.

}

\section{Introduction} \label{sec: intro}

In experiments where the assumption of ``no interference between units'' does not apply, the outcome of each unit may depend not only on their own treatment assignment, but also on the treatment of others. Examples include vaccination studies where units receiving the placebo may be protected by the vaccine through herd immunity \citep{hudgens2008toward}, or social networks in which units may be influenced by the actions of others \citep{cai2021causal}. In such experiments, often an assumption is made to limit the effects of interference; for example, one might divide the units into groups that are assumed not to interfere with each other, or assume that most units do not interfere with each other. However, such assumptions may sometimes be implausible or overly questionable, precluding their usage.

For such settings, we propose new estimands that are more limited in their causal interpretation, but can be interval estimated with no assumptions on interference. These estimands correspond to contrasts in attributable effects. The attributable effect was proposed in \cite{rosenbaum2001effects}. At the unit level, it is the difference between each unit's observed outcome, and their outcome under a prespecified counterfactual scenario, such as the counterfactual in which all units receive the control. The intuition for this effect is the following: if a unit had a certain outcome under the observed treatment assignment, but would have had a different outcome had all units received the control, then colloquially one could say that under the observed assignment, the unit was affected by the treatments. Since the attributable effect depends on the observed treatment assignment, it is not a parameter, but rather a latent random variable.

We propose to partially characterize the unit-level attributable effect by taking contrasts -- that is, by comparing its values for units with different levels of direct treatment, or different levels of indirect exposure to the treatment of others. Such contrasts may occur as non-intercept terms in a regression, such as the regression of the attributable effect against treatment,
\begin{equation}\label{eq: intro regression}
  \eff_i \sim \beta_0 + \beta_1 \cdot \textit{Treatment}_i
\end{equation}
where $\eff_i$ is the attributable effect for unit $i$, and $\textit{Treatment}_i$ is their observed treatment assignment. In this regression, $(\beta_0,\beta_1)$ are given by
\begin{equation}\label{eq: intro least squares}
  \beta_0,\beta_1 = \arg \min_{\beta_0,\beta_1} \sum_{i=1}^N \left(\eff_i - \beta_0 - \beta_1 \cdot \textit{Treatment}_i\right)^2
\end{equation}
and $\beta_1$ compares the attributable effects of the treated and control. As a more complex example, we might regress effects against each unit's treatment, their number of treated friends, and an interaction term,
\begin{align}\label{eq: intro regression 2}
  \eff_i \sim \beta_0 + \beta_1 \cdot \textit{Treatment}_i + \beta_2 \cdot \textit{\#Treated Friends}_i + \\ 
  \nonumber \beta_3 \cdot (\textit{Treatment}_i \cdot \textit{\#Treated Friends}_i)
\end{align}
In both \eqref{eq: intro regression} and \eqref{eq: intro regression 2}, $\eff$ is not observed, and hence statistical inference is required to estimate the regression parameters.

The main result of this paper is that in regression specifications such as \eqref{eq: intro regression} and \eqref{eq: intro regression 2} with binary outcomes, contrasts such as $\beta_1$, $\beta_2$, and $\beta_3$ (but not the intercept term $\beta_0$) can be interval estimated with no assumptions beyond those regarding randomization of treatment. Thus, these parameters can be reliably inferred even when interference between individuals is difficult to accurately model, including settings in which an observed network only crudely reflects the true underlying social mechanisms. The approach can be extended to other regression specifications, including settings where exposure to treatment is non-identically distributed. It also generalizes beyond regression to weighted or matching-based comparisons of units. Novel computations are required to bound the interval widths. 

In the regression \eqref{eq: intro regression}, the coefficient $\beta_1$ measures the relative attributable effect of being treated versus not treated. To give an example, consider a vaccination trial where those receiving the placebo might be indirectly affected by the vaccines, with the outcome being incidence of disease. In this setting, estimates of $\beta_1$ could be used to show that vaccinated units were relatively more protected than those who received the placebo, by an additional $\beta_1$ cases prevented (or fewer cases caused) per capita, under the observed treatment assignment. Similarly, in \eqref{eq: intro regression 2}, estimates of $\beta_1$, $\beta_2$, and $\beta_3$ could be used to show that the relative attributable effects were consistent with those of a working vaccine that imbued herd immunity: under the observed treatment assignment, vaccinated units were relatively more protected than placebo units, while placebo units were relatively more protected when friends were vaccinated.

In the absence of assumptions on interference, hypothesis testing can be used to reject Fisher's sharp null of no effect, as well as analogously defined nulls of no interference and similar extensions \citep{athey2017exact, aronow2012general, basse2019randomization, pouget2017testing}. However, without further assumptions these tests do not imply interval estimates for the treatment effect. For example, rejection of the sharp null of no effect shows only that at least one unit was affected. This leaves open that the effects may have been small, and that a difference in average outcomes for the treated and control may have been primarily due to unmeasured pre-treatment differences between the two groups. It even allows that the effect may have been in the opposite direction from that suggested by the observations. In contrast, our estimands can establish whether (and to what extent) a difference in observed outcomes can be attributed to the effects of the treatments. Additionally, they answer whether the effects varied systematically as a function of each unit's direct and indirect exposure to treatment. Finally, our estimands will often imply better lower bounds on the number of affected units. 

\paragraph{Limitations}





To illustrate the difficulties that can arise when inteference is completely unrestricted, imagine a stylized example in which a colony of ants is subjected to a randomized trial. (Ant colonies are highly communal, suggesting that high levels of interference may be possible.) Treated ants are exposed to a drug, while all ants are marked in some innocuous manner to identify their treatment assignment. After some time, survival rates are recorded for the treated and control groups. 

We consider two difficulties that may arise in our fictional setting. First, it can be seen that even if the treated ants had a significantly higher survival rate than the control group, the overall effects of the treatment may still have been harmful for all units. For example, if treated ants became stronger but also ceased activities necessary for colony survival, such as foraging for food, then all ants might have lower survival rates, with the lowest rates for the control group. Second, it can be seen that the survival rates need not converge to their expectation. For example, suppose that the colony collapses if a particular ant is treated (whose identity is unknown to the researcher), and functions normally otherwise. If this ant is treated with probability $p$, then the colony collapses with probability $p$ and functions normally otherwise, regardless of the experiment size and randomization of other units. 

The implications of these difficulties for our estimands are the following. First, without the intercept term $\beta_0$, which is not identified without additional assumptions, we can only infer relative differences in the attributable effect. This leaves open the possibility that ``more protected'' may mean ``less harmed,'' or that a treatment that benefits the individual may have negative global consequences, causing everyone to become worse off. Second, the attributable effect is a random variable, and our inference is not for its distribution or expectation, but rather for its latent value. This means our confidence statements will not address how the units would be affected under a new treatment assignment, but only how they {\it were} affected under the observed one. 

Due to these limitations, our inference does not result in a confidence statement for the effects of a hypothetical intervention. Such statements cannot be made without additional assumptions on interference. Instead, our intervals are for a partial characterization of the effects that occurred under the observed treatment assignment. This stops short of giving a definitive policy recommendation, but may still be of interest to researchers and decision makers if stronger inferences cannot be credibly made.

\subsection{Related Work} \label{sec: related work}

Our work is part of a growing literature on statistical inference in the presence of interference between units. Approaches vary in their target estimand and severity of modeling assumptions, with stronger inferences generally requiring stronger assumptions.

Under the assumption of a known exposure model \citep{aronow2017estimating, eckles2017design}, the outcome of each unit depends on the treatments through a set of statistics, termed the unit exposure; for example, a unit's exposure might equal the treatment of its neighbors in an ellicited social network. In this setting, average treatment effects are generalizable to the difference in outcomes between any two treatment allocations (such as full treatment and full control).

A growing body of work considers the expected effects of intervening on the treatment of a single unit, assuming the treatment of all other units to be randomized according to the experiment design  \citep{hudgens2008toward, savje2017average, leung2019causal, liu2014large, hu2021average, li2020random}. For this class of estimands, unbiased point estimation requires no assumptions beyond randomization; however, variance estimates and valid confidence intervals require assumptions on interference, although not necessarily on its exact structure. \cite{leung2019causal} requires a bandwidth parameter to be chosen according to the strength of the assumed interference;  \cite{savje2017average} proposes variance inflation factors based on the degrees of the underlying dependency graph; and \cite{li2020random} consider various assumptions in conjunction with distributional interactions on a random network.

With the exception of Appendix \ref{sec: cholera real}, we will make no assumptions on interference, so that it may be unknown and unbounded in strength. Previous work in this setting includes testing for null hypotheses related to peer effects, such as the hypothesis of direct effects but no spillovers \citep{athey2017exact, aronow2012general, basse2019randomization, pouget2017testing}, and estimation of attributable effects on the ranks of the units when ordering by outcome \citep{rosenbaum2007interference}. 

\paragraph{Other related works} Estimation under parametric modeling assumptions on interference is studied in \cite{loh2018randomization} and \cite{bowers2013reasoning}, and under an assumption that the effect is nonnegative in \cite{choi2017estimation}. \cite{cortez2022staggered} estimate total effects assuming bounds on the extent to which treatments can interact nonlinearly. \cite{cai2021causal} discusses vaccine-specific effects. Inference in observational studies with interference is studied in  \cite{forastiere2020identification} and \cite{ogburn2017causal}. 
\cite{rigdon2015exact} use attributable effects for binary outcomes, but in the no interference setting. A technical report \citep{choi2018manuscript} considers the estimand $\tau_1$ given by \eqref{eq: tau1}, but without generalization to other regressions, weighting, or matching. 

\section{Basic Approach} \label{sec: basic}

To introduce the main concepts, and to draw parallels with existing approaches, we first describe our method for the contrast arising in the simple regression \eqref{eq: intro regression} that measures differences in attributable effects for the treated and control. 


\subsection{Defining the Attributable Effect} \label{sec: setup}

We assume a randomized experiment on $N$ units, and let $X \in \{0,1\}^N$ denote the treatment assignment (whose distribution is known), and $Y\in \{0,1\}^N$ the observed outcomes. We assume unrestricted inteference, so that each unit's outcome depends potentially on all $N$ treatments:
\begin{equation}\label{eq: interference}
Y_i = f_i(X_1,\ldots,X_N),
\end{equation}
where the outcome mappings $\{f_i\}$ are arbitrary and unknown. 

Let $\theta\in \{0,1\}^N$ denote the counterfactual outcomes that would result under a ``uniformity trial'' \citep{rosenbaum2007interference}. A uniformity trial may be any counterfactual in which the treatments are not administered according to $X$, and hence the outcomes are unaffected by it. A common choice for $\theta$ is the counterfactual of ``full control'', in which all units receive the control and $\theta$ is given by
\[ \theta_i = f_i(0,\ldots,0). \] 
We typically make no assumptions on $\theta$ except on its range of allowable values. In particular, we do not assume the data to be informative for the values of $\theta$.

Let $A \in \{-1,0,1\}^N$ denote the {\it individual-level attributable treatment effect}, with entries given by
\[ A_i = Y_i - \theta_i, \]
which was denoted by $\eff_i$ in Section \ref{sec: intro}. For example, in a vaccination trial the treatment $X_i$ might denote whether unit $i$ received the vaccine or a placebo, $Y_i$ whether the unit had the disease, and $\theta_i$ the counterfactual outcome that would have resulted if all units were given the placebo. Letting an outcome of $1$ denote the disease and $0$ otherwise, an effect of $A_i = -1$ would imply that $Y_i=0$ and $\theta_i=1$, meaning that unit $i$ did not have the disease, but would have if the placebo had been administered to all participants. In such a case, it can be said that administering the vaccine prevented this unit from having the disease. This can occur even if the unit received the placebo in the trial, due to social effects such as herd immunity. Thus $A$ clearly has a causal interpretation. 

\subsection{A Simple Estimand}

Our goal is to partially characterize $A$ by estimating a contrast term appearing in a regression, or in a weighted or matching-based comparison of the unit outcomes. Let $N_1$ and $N_0$ denote the number of treated and control units, and let $\bar{X} = N_1/N$ denote the average treatment. As a simple example, we consider the term $\tau_1$ appearing in the following regression:
\begin{equation*}\label{eq: tau regression}
  A_i \sim \tau_0 + \tau_1 (X_i - \bar{X}),
\end{equation*}
which is a slight modification of \eqref{eq: intro regression}.
It can be seen that the intercept $\tau_0$ and slope $\tau_1$ are given by
\begin{align} 
\tau_0 &= \frac{1}{N} \sum_{i=1}^N A_i \label{eq: tau0} \\
\tau_1 &= \frac{1}{N_1} \sum_{i:X_i = 1} A_i - \frac{1}{N_0} \sum_{i:X_i=0} A_i, \label{eq: tau1}
\end{align}
so that $\tau_0$ is the average attributable effect, and $\tau_1$ is the difference in attributable effects for the treated and control. Thus $\tau_1$ answers whether (and to what extent) the units receiving the treatment were affected differently than those receiving the control, under the observed treatment assignment. 

\paragraph{$\mathbf{\tau_0}$ Cannot be Estimated} Without additional assumptions, the variable $\tau_0$ cannot be estimated. To see this, consider that
\[ \tau_0 = \bar{Y} - \bar{\theta},\]
where $\bar{Y}$ and $\bar{\theta}$ denote the averages of $\{Y_i\}_{i=1}^N$ and $\{\theta_i\}_{i=1}^N$. Since $\theta$ is unobserved and $Y$ may be uninformative for the values of $\theta$, without further assumptions we can only say that $\bar{\theta}$ may be any value in $[0,\,1]$, and hence that $\tau_0 \in [\bar{Y} - 1, \, \bar{Y}]$.

\paragraph{Estimation of $\tau_1$}  
 Conservative prediction intervals (PIs) can be found for the unobserved random variable $\tau_1$, when $X$ is assigned by sampling without replacement. Let $\hat{\tau}_1$ denote a point estimate for $\tau_1$ equal to the difference in observed outcomes,
  \begin{equation} \label{eq: hat tau 1}
  \hat{\tau}_1 = \frac{1}{N_1} \sum_{i:X_i = 1} Y_i - \frac{1}{N_0} \sum_{i:X_i=0} Y_i.
  \end{equation}
It can be seen that the estimation error $\hat{\tau}_1 - \tau_1$ is given by
\begin{align} \label{eq: tau error}
 \hat{\tau}_1 - \tau_1 & = \sum_{i:X_i=1} \frac{\theta_i}{N_1}  - \sum_{i:X_i=0} \frac{\theta_i}{N_0},
\end{align}
which is the difference between a random sample drawn from $\{\theta_i\}_{i=1}^N$ and its complement. As a result, we can apply a finite population central limit theorem to find a prediction interval for $(\hat{\tau}_1 - \tau_1)$, and hence for the estimand $\tau_1$. The width of this interval will depend on the unknown variance $\sigma_\theta^2$, given by
  \begin{align} \label{eq: theta mean and var}
  \bar{\theta} & = \frac{1}{N} \sum_{i=1}^N \theta_i, & \sigma_\theta^2 & = \frac{1}{N} \sum_{i=1}^n (\theta_i - \bar{\theta})^2.
  \end{align}
Proposition \ref{th: CLT tau} gives the prediction interval, using the bound $\sigma_\theta^2 \leq 1/4$ which holds for binary outcomes.\footnote{In the terminology of \cite{imbens2004confidence}, $\sigma_\theta^2$ has identification set $[0,\, 1/4]$, inducing an identification set for the prediction interval of $\hat{\tau}_1 - \tau_1$.}  It is proven in Section \ref{sec: CLT tau proof} of the supplement. 

\begin{proposition} \label{th: CLT tau}
  Given $\rho \in (0,1)$, assume a sequence of experiments in which $N \rightarrow \infty$, where treatment is assigned by sampling $N_1$ units without replacement; $N_1/N$ converges to $\rho$; and $\theta \in \{0,1\}^N$.   Then $\tau_1$ is contained by the interval
  \begin{align*}
    \hat{\tau}_1 \pm z_{1-\alpha} \sqrt{\frac{N}{N-1}\frac{N}{N_1 N_0}\frac{1}{4}}
  \end{align*}
with probability asymptotically lower bounded by $1-2\alpha$.

  %
\end{proposition}

\paragraph{Number of Affected Units}

It can be seen for binary outcomes that 
\begin{equation}\label{eq: tau number of affected units}
\sum_{i=1}^N 1\{Y_i \neq \theta_i\}  \geq |\tau_1|\cdot \min(N_1, N_0),
\end{equation}
as proven in the supplement using H\"older's inequality. Thus $|\tau_1|\cdot \min(N_1,N_0)$ lower bounds the number of affected units under the observed treatment assignment, in the sense that their outcomes were changed from what would have occurred under the uniformity trial. It also lower bounds the number of units affected by treatment in the sense that their potential outcome mapping $f_i$ is not a constant.

\paragraph{Distribution of $\mathbf{\tau_1}$ cannot be estimated} To give a negative example, suppose that the outcome of every unit depends on the treatment of unit 1, according to 
\begin{equation}\label{eq: high influence unit}
   Y_i = X_1 g_i(X_i) + (1-X_1) h_i(X_i).
\end{equation}
  The distribution of $\tau_1$ cannot be estimated from an experiment -- it depends on both $g$ and $h$, but only one of the two will be observed.

\subsection{Discussion} \label{sec: tau discussion}


Our inference is for the latent value of $\tau_1$ under the observed treatment assignment, and not for its distribution or expectation. As a result, we don't know whether the value of $\tau_1$ is sensitive to random perturbation of the treatment assignment, or equivalently whether its value would replicate if the experiment was repeated on an identical set of units with a new draw of the treatment assignment. 

This limitation is inherent to the setting that we study, as illustrated by \eqref{eq: high influence unit} as an extreme case. Here the treatment 
of a particular unit affected all others, and hence the experimental outcomes could vary significantly with that unit's treatment, regardless of the experiment size. Thus non-replicability of experiments may arise not only from differences between populations, or differences between experimental designs, but also from intrinsic variability caused by the random treatment of high-influence units.

This suggests that in the absence of assumptions on interference, multiple experiments may be required to be confident that the effects of the treatment are not sensitive to random perturbation of the treatment assignment. The estimand $\tau_1$ provides an intermediate finding that can be learned from a single experiment.  It lower bounds the number of affected units, and quantifies the extent to which an observed difference in outcomes can be attributed to the effects of the treatment, as opposed to being attributed to pre-treatment differences between units. This may help to determine whether further experiments are of interest. 

Knowing that the units of an experiment were affected in a certain way will be of interest to a decision maker considering the same policy. To communicate the limited nature of $\tau_1$ in such settings, one might state that in the absence of assumptions on interference, no statistical guarantee can be made regarding the replicability of a single experiment -- even for an identical population -- due to the unknown variability of the system under study. This is consistent with general wisdom regarding studies of complex phenomena, such as the recommendation by \cite[p. 161]{national2019reproducibility} that
\begin{quote}
  Anyone making personal or policy decisions based on scientific evidence should be wary of making a serious decision based on the results, no matter how promising, of a single study. 
\end{quote}
If a decision must be based on a single experiment, the limitations of $\tau_1$ may help to convey the need for caution -- for example, by placing more importance on safeguards, such as monitoring for adverse events.





\paragraph{Example}
We give an example where the relevance of $\tau_1$ is competitive with estimands whose inference requires stronger assumptions. Suppose that individuals were randomly selected for military service through a draft lottery. The drafted individuals then fight in a war, whose outcome affects all units. Interference seems plausible, as the drafted units determine the outcome of the war, and additionally they might form social ties that persist afterwards. An important question might be whether drafted units were affected more negatively than those who were undrafted. If so, we might allocate funds or create policies to help them.  For example, \cite{angrist1990lifetime} considers the effects of the Vietnam War draft lottery on lifetime earnings. For this purpose, it is sufficient to estimate the latent value of $\tau_1$. Knowledge of its expectation or distribution would not help further, as the question at hand is not to consider a hypothetical intervention in which the units are drafted again, but rather how we should react to the draft that actually occurred.

\subsection{Comparison with other Estimands}


\paragraph{SATE and SATT}

In the absence of interference, each unit's outcome depends only on their own treatment,
\begin{equation}\label{eq: sutva}
 Y_i = f_i(X_i),
\end{equation}
and the sample average treatment effect (SATE) is defined to be
\[ \textup{SATE} = \frac{1}{N} \sum_{i=1}^N (f_i(1) - f_i(0)),\]
the difference in outcomes between treatment and control, averaged over the units. Thus the SATE gives the total effects of the treatment. In contrast, $\tau_1$ gives only relative effects, and only under the observed treatment assignment. In settings where the decision is whether to distribute the treatments more broadly, information provided by the SATE is far more actionable, provided that \eqref{eq: sutva} holds.

Under \eqref{eq: sutva}, and for $\theta$ equal to the counterfactual of full control, $\tau_1$ equals the sample average treatment effect on the treated (SATT), which is given by
\[ \textup{SATT} = \frac{1}{N_1} \sum_{i:X_i=1} (f_i(1) - f_i(0)) \]
The SATT gives the effect on intervening on the units that received treatment. It is a random variable whose value may differ from the SATE, particularly in observational studies where certain units may be more likely to receive treatment.






\paragraph{Comparison of Prediction Intervals} 

The prediction intervals for SATT and $\tau_1$ are closely related. Letting $\bar{\theta}_x$ for $x \in \{0,1\}$ denote the average of $\{\theta_i: X_i = x\}$ and defining $\bar{Y}_x$ analogously, the estimation error for the SATT is given by $\hat{\tau}_1 - \textup{SATT} = \bar{\theta}_0 - \bar{\theta}_1$, where we have used \eqref{eq: sutva} and $\bar{Y}_0 = \bar{\theta}_0$. This error is zero mean and has variance \citep{sekhon2021inference}
\begin{equation}\label{eq: var ATT}
 \operatorname{Var}\left(\hat{\tau}_1 - \textup{SATT}\right) = \frac{N}{N-1}\frac{N}{N_1N_0} \sigma_\theta^2,
\end{equation}
with $\sigma_\theta^2$ given by \eqref{eq: theta mean and var}. Since $Y_i = \theta_i$ for the control units, the empirical variance of the control outcomes $\{Y_i:X_i = 0\}$ is consistent for $\sigma_\theta^2$, allowing plug-in estimation of \eqref{eq: var ATT}.

In the presence of interference, $\theta_i$ need not equal $Y_i$ for the control units. In such settings we can estimate $\tau_1$. As shown by \eqref{eq: tau error}, the estimation error $\hat{\tau}_1 - \tau_1$ again equals $\bar{\theta}_1 - \bar{\theta}_0$, and hence its variance is the same as \eqref{eq: var ATT}. 
However, $\sigma_\theta^2$ can no longer be estimated from the control outcomes, and must be bounded by 1/4 if no further assumptions are made, resulting in a prediction interval that is wider than that of the $\textup{SATT}$ by a factor of $(2\sigma_{\theta})^{-1}$. 

In some cases, it may be reasonable to reduce this factor by bounding $\sigma_\theta^2 = \bar{\theta}(1-\bar{\theta})$ more tightly. For example, in the supplement we consider a vaccination trial in which the disease rate ranged from 0.002 to 0.005 in the preceding years. In such a case, it might reasonable to bound $\bar{\theta}$ by a safety factor of the historical values, as opposed to allowing $\bar{\theta} = 0.5$, which would be implausibly high for this setting.

\paragraph{EATE}
The expected average treatment effect (EATE), which is also termed the direct effect, is equal to
\begin{equation}
 \textup{EATE} = \frac{1}{N} \sum_{i=1}^N \left(\mathbb{E}[Y_i| X_i = 1] - \mathbb{E}[Y_i | X_i=0]\right).
\end{equation}
In settings where the units are independently assigned to treatment or control, the EATE is the expected effect of intervening on a single unit, where the effect is measured only on the outcome of that unit, and the expectation is over the randomly assigned treatment of the remaining units. Similar to $\tau_1$, the EATE does not capture total effects -- a treatment assignment might have positive direct effects but negative global effects, so that all units are made worse off.

Unlike $\tau_1$, the EATE is a parameter and not a random variable. As a result, its implications are not limited to the particular treatment assignment that was observed. However, at least nominally the EATE considers only interventions on a single unit. Thus for interventions involving many units the EATE may have more limited relevance.

The EATE and $\tau_1$ have different sources of estimation error, which are not ``nested'' in any sense. Unlike $\tau_1$, the estimation error for the EATE depends on whether $\hat{\tau}_1$ concentrates to its expectation, and requires additional assumptions to bound.\footnote{Under reasonable designs (such as that of Proposition \ref{th: CLT tau}) it holds that $\hat{\tau}_1$ is consistent for $\tau_1$. For such settings, consistency of $\hat{\tau}_1$ for the EATE will also imply convergence of $\tau_1$ to the EATE.} As a result, the interval width for the EATE will depend on the strength of the assumption that is made. Under strong assumptions, the interval width for the EATE may approach those attainable under an assumption of no interference; under very weak assumptions, consistent estimation of the EATE may not be possible.

\paragraph{Examples of Structural Models}

We compare the EATE and $\tau_1$ for two stylized examples of structural models. For simplicity, we assume continuous-valued outcomes and i.i.d. Bernoulli treatment randomization, and let $\theta$ denote the counterfactual outcomes under full control.

\begin{enumerate}

\item Given an observed social network, let $Z_i$ denote unit $i$'s number of neighbors assigned to the treatment group, and consider the model
\[ Y_i = \alpha_{0i} + \alpha_{1i} X_i + \alpha_{2i} Z_i,\]
in which case $\theta_i$ and $A_i$ are given by
\begin{align*}
\theta_i & = \alpha_{0i} & A_i = \alpha_{1i}X_i + \alpha_{2i} Z_i, 
\end{align*}
and $\tau_1$ equals
\[ \tau_1 = \frac{1}{N_1} \sum_{i:X_i=1} \alpha_{1i} + \left(\frac{1}{N_1} \sum_{i:X_i=1} \alpha_{2i}Z_i - \frac{1}{N_0} \sum_{i:X_i=0} \alpha_{2i}Z_i \right).\]
Thus $\tau_1$ for this model equals the average of direct effects for the treated units, plus the difference in spillover effects for the treated and control under the observed assignment. The EATE for this setting is the average of the direct effects for all units, as given by
\begin{equation}\label{eq: EATE example}
 \textup{EATE} = \frac{1}{N} \sum_{i=1}^N \alpha_{1i}
\end{equation}
For large $N$, it can be seen that $\tau_1$ converges to the EATE.


\item For the model of \eqref{eq: high influence unit} in which all units are affected by the treatment of unit 1, it can be seen that $\theta_i = h_i(0)$. As a result, $\tau_1$ can be seen to equal
\[ \tau_1 = \begin{cases} \displaystyle \sum_{i:X_i=1} \frac{g_i(1) - h_i(0)}{N_1} - \sum_{i:X_i=0} \frac{g_i(0) - h_i(0)}{N_0} & \textup{ if $X_1=1$} \\
\displaystyle \sum_{i:X_i=1} \frac{h_i(1) - h_i(0)}{N_1} & \textup{ if $X_1=0$,} \end{cases} \]
the difference in effects under the observed treatment assignment.
This does not converge to the EATE, which is given by 
\begin{align*}
 \textup{EATE} & = \frac{g_1(1) - h_1(0)}{N} + p \sum_{i=2}^N \frac{g_i(1) - g_i(0)}{N} \\
 & \qquad + (1-p) \sum_{i=2}^N \frac{h_i(1) - h_i(0)}{N},
 \end{align*}
where $p=P(X_1=1)$.

\end{enumerate}

\section{General Case} \label{sec: general case}

We generalize $\tau_1$ to allow the estimand to be a linear function of $A$, 
\begin{equation}\label{eq: general estimand}
\text{estimand} = \sum_{i=1}^N w_i(X) A_i,
\end{equation}
where our prediction intervals for the estimand will be smallest when the weight vector $w\equiv w(X)$ is approximately zero-mean in each entry. To implement this practical requirement when the treatment or number of treated neighbors is non-identically distributed, we can use matching, weighting, or regression adjustments to draw approximately unbiased comparisons between units.

\paragraph{Preliminaries} 

To measure each unit's indirect exposure to the treatment of others, let $G \in \{0,1\}^{N \times N}$ denote the adjacency matrix of an observed network, let $Z = (Z_1,\ldots,Z_N)$ denote the number of treated neighbors for each unit,
\begin{equation}\label{eq: Z}
Z_i = \sum_{j: G_{ij} = 1} X_j,
\end{equation}
and $W = (W_1,\ldots,W_N)$ the thresholded exposures at a fixed level $z_{\textup{min}}$:
\begin{equation}\label{eq: W}
 W_i = 1\{Z_i \geq z_{\textup{min}}\}.
\end{equation}
We will assume that the units can be partitioned into $K$ classes $\Pi_1,\ldots,\Pi_K$,  where $n_k = |\Pi_k|$ denotes the size of class $k$, and $k(i)$ equals the class of unit $i$. Our prediction intervals will have valid coverage even if the classes are chosen arbitrarily; however, in practice this is not recommended as the intervals may become very wide. Ideally, the classes should approximate propensity classes, so that for two units $i$ and $j$ belonging to the same class, $(X_i, Z_i)$ and $(X_j, Z_j)$ will be approximately identically distributed. For example, if $X$ is assigned by sampling without replacement, then the distribution of $Z_i$ varies only with $i$'s number of neighbors in the network $G$. In this case, grouping the units by their number of neighbors results in identically distributed exposures in each class.


\paragraph{Regression}

Let $\beta \in \mathbb{R}^J$ denote a general regression specification, solving the least squares problem
    \[ \min_{\beta} \sum_{i=1}^N \left(A_i - v_i(X)^T \beta\right)^2,\]
    for some choice of regressors $v_i(X) \in \mathbb{R}^J$ for $i \in [N]$. Then $\beta$ is given by
    \begin{equation} \label{eq: general regression beta}
     \beta = \left(\sum_{i=1}^N v_i(X) v_i(X)^T\right)^{-1} \sum_{i=1}^N v_i(X)A_i,
     \end{equation}
    and the $\ell$th entry of $\beta$ can be written in the form \eqref{eq: general estimand} for $w$ given by
    \begin{align} \label{eq: general regression w}
     w_i & = e_\ell^T \left(\sum_{i=1}^N v_i(X) v_i(X)^T\right)^{-1} v_i(X),
     \end{align}
     where $e_\ell \in \mathbb{R}^J$ is the $\ell$th column of the $J \times J$ identity matrix. 

To give an example, let $\beta^\tadj$ denote the partial regression coefficient that arises when regressing the attributable effect $A$ on the exposures $Z$ and the set of $K$ indicator variables representing the membership of each unit in the propensity classes $\Pi_1,\ldots,\Pi_K$:
\begin{equation}\label{eq: beta adj}
\beta^{\text{adj}} = \left(\sum_{i=1}^N (Z_i - \zeta_{k(i)})^2\right)^{-1} \sum_{i=1}^N (Z_i -\zeta_{k(i)})A_i,
\end{equation}
where $\zeta_k$ is the average exposure for the $k$th propensity class
\[ \zeta_k = \frac{1}{n_k} \sum_{i \in \Pi_k} Z_i\]
To interpret $\beta^{\tadj}$, we may say that for two individuals in the same propensity class, a unit difference in the level of exposure is associated with a $\beta^\tadj$ difference in their individual-level attributable effects. 

As a second example, we might regress on a set of indicator variables for the possible values of $Z$, plus the propensity classes. This results in the least squares problem
\begin{equation}\label{eq: effect curve}
 \min_{\gamma \in \mathbb{R}^{d_{\max}}, \nu \in \mathbb{R}^K} \sum_{i=1}^N \left(\gamma_{Z_i} + \nu_{k(i)} - A_i\right)^2,
\end{equation}
where $d_{\max}$ denotes the maximum degree of the network $G$. Here $\gamma$ corresponds to the levels of $Z$, and $\nu$ corresponds to the propensity classes. We may then say that $\gamma_z - \gamma_0$ predicts the difference in attributable effects for a unit with exposure level $z$ compared to one in the same propensity class who had zero exposure.

\paragraph{Weighting and Matching}

Given thresholded exposures $W$ and propensity classes $\Pi_1,\ldots,\Pi_K$, let $n_{k1}$ and $n_{k0}$ denote the number of units in $\Pi_k$ for which $W_i=1$ and $W_i=0$, respectively. Let $\tau^{\tweight}$ denote the weighted difference
\begin{equation}\label{eq: tau weighted}
 \tau^{\tweight} = \sum_{k=1}^K \frac{n_k}{N} \sum_{i \in \Pi_k}\left(  \frac{W_iA_i}{n_{k1}} - \frac{(1-W_i)A_i}{n_{k0}}\right),
\end{equation}
which can be interpreted as the difference in attributable effects within each propensity class, averaged together after weighting by class size. It can be seen that $\tau^{\tweight}$ can be written in the form $w^TA$ for $w$ given by
  \begin{equation} \label{eq: weighted w}
     w_i = \frac{n_{k(i)}}{N} \left(\frac{W_i}{n_{k(i)1}} - \frac{1-W_i}{n_{k(i)0}}\right) 
  \end{equation}

Alternatively, we can construct a matched subset of units by randomly pairing units within the same propensity class (where one unit in each pair has $W_i=1$ and the other has $W_i=0$), without replacement until no more pairs can be formed. Let $\mathcal{M}$ denote the set of all units included in the matching, and let $m = |\mathcal{M}|/2$ denote the number of matched pairs. Let $\tau^\tmatch$ denote the difference in attributable treatment effects between the matched units, 
\begin{equation} \label{eq: tau matched}
\tau^\tmatch = \frac{1}{m} \left(\sum_{i \in \mathcal{M}:W_i=1} A_i - \sum_{i \in \mathcal{M}:W_i=0} A_i\right).
\end{equation}
The value $\tau^\tmatch$ is analogous to $\tau$, measuring the difference in attributable effects for the units in the matching. It can be seen that $\tau^\tmatch$ equals $w^TA$ for $w$ given by
  \begin{equation} \label{eq: matched w}
     w_i = \begin{cases} \frac{1}{m} & \text{ if } i \in \mathcal{M}, W_i = 1 \\
     -\frac{1}{m} & \text{ if } i \in \mathcal{M}, W_i = 0 \\
     0 & \text{ otherwise} \end{cases}
  \end{equation}

As units are randomly paired within the same propensity class, the quantity $\tau^\tmatch$ is random given $X$. If this is undesireable, the expectation of $\tau^\tmatch$ may be used instead. This quantity closely resembles $\tau^\tweight$, and is given by
\begin{equation}\label{eq: E tau matched}
 \E[\tau^\tmatch|X] = \sum_{k=1}^K \frac{m_k}{m} \sum_{i \in \Pi_k} \left(\frac{W_iA_i}{n_{k1}} - \frac{(1-W_i)A_i}{n_{k0}}\right),
\end{equation}
where $m_k = \min(n_{k1}, n_{k0})$ equals the number of matched pairs formed in the $k$th strata, with $m = \sum_k m_k$.

\paragraph{Number of Affected Units}

Analogous to \eqref{eq: tau number of affected units}, the estimand $w^TA$ implies the following lower bound on the number of units affected,
\begin{equation}\label{eq: w number of affected units}
\sum_{i:w_i \neq 0} 1\{Y_i \neq \theta_i\} \geq \frac{|w^TA|}{\max_{i \in [N]}|w_i|},
\end{equation}
as proven in the supplement.

\subsection{Estimation} In principle, estimation for the general estimand $w^TA$ proceeds similarly to that of $\tau_1$, by defining the point estimate to equal $w^T Y$ and then bounding the distribution of the estimation error.

This point estimate will usually take intuitive form. For example, if the estimand $w^TA$ is a regression on the attributable effects, then $w^TY$ is simply the regression applied to the observed outcomes. Likewise, if $w^TA$ is a weighted or matching-based comparison, then $w^TY$ is simply that comparison applied to the observed outcomes. 

This point estimate need not be unbiased, in the sense that the estimation error may not be zero-mean; however, the expectation of the error (i.e., bias) can be bounded by 
\begin{equation}\label{eq: bias bound}
   \min_{\theta \in \{0,1\}^N} \mathbb{E}w^T\theta\ \leq \mathbb{E}\left[w^TY - w^TA\right] \leq \ \max_{\theta \in \{0,1\}^N} \mathbb{E}w^T\theta
\end{equation}
or equivalently by 
\begin{equation}\label{eq: bias bound 2}
   \sum_{i=1}^N \min(\mathbb{E}w_i, 0) \ \leq \mathbb{E}\left[w^TY - w^TA\right] \leq \ \sum_{i=1}^N \max(\mathbb{E}w_i, 0).
\end{equation}

To find a prediction interval for $w^T A$, we require a central limit theorem to hold for $w^T\theta$, so that
\[ \frac{w^T\theta - \mathbb{E}w^T\theta}{\sqrt{\operatorname{Var}(w^T\theta)}} \rightarrow N(0,1),\]
implying a normal-based interval for $w^T\theta$ given by
\[ \mathbb{E}w^T\theta \pm z_{1-\alpha} \sqrt{\operatorname{Var} [w^T\theta}].\]
As $\theta$ is unobserved, the boundaries of this interval are unknown, but we can find an outer bound by maximizing and minimizing over the unknown values of $\theta$, resulting in the upper and lower bounds $U$ and $L$ given by
\begin{align} 
\label{eq: U} U & = \max_{\theta \in \{0,1\}^N}\ \mathbb{E}w^T\theta + z_{1-{\alpha}} \sqrt{\operatorname{Var}[w^T\theta]} \\
\label{eq: L} L & = \min_{\theta \in \{0,1\}^N}\ \mathbb{E}w^T\theta - z_{1-{\alpha}} \sqrt{\operatorname{Var}[w^T\theta]},
\end{align}
and a $(1-2\alpha)$ level prediction interval for the estimand is given by $w^TY - [U,\, L]$.\footnote{If $-U \leq -w^T\theta \leq -L$ holds with probability $1-\alpha$, then adding $w^TY$ implies $w^TY -U \leq w^TA \leq w^TY -L$ holds with the same probability.} 

Appendix \ref{sec: asymptotics} gives results showing consistency and interval coverage for regression-based estimands.

\subsection{Computation} \label{sec: computation}

As $\operatorname{Var}(w^T\theta)$ equals $\theta^T Q \theta$ for $Q = \mathbb{E}ww^T - \mathbb{E}w\mathbb{E}w^T$, solving for $U$ given by \eqref{eq: U} requires solving the optimization problem
\begin{equation}\label{eq: optimization}
 \max_{\theta \in \{0,1\}^N} \mathbb{E}w^T\theta + z_{1-\alpha} \sqrt{\theta^T Q \theta},
\end{equation}
To ensure coverage of the prediction intervals, a local optimum to \eqref{eq: optimization} is not sufficient -- either the global solution or an upper bound must be found. As $Q$ is positive semi-definite, this requires maximization of a convex objective, which is computationally difficult even if the integer constraint on $\theta$ is relaxed to the convex set $[0,1]^N$. 

To solve or upper bound \eqref{eq: optimization}, we proceed in two steps. First, we find a diagonal matrix $D$ such that $Q - D$ is negative semidefinite. This can be done by solving the semidefinite program 
\begin{align*} \label{eq: sdp}
  \min_{D \in \mathbb{R}^{n\times n}} \sum_{i=1}^n D_{ii} \qquad \textup{subject to } D \textup{ diagonal, $D \geq 0,$ and } Q-D \preceq 0 
\end{align*}
where $\geq 0$ denotes nonnegativity and $\preceq 0$ denotes negative semidefinite. Then, we can rewrite \eqref{eq: optimization} as  
\begin{equation}\label{eq: optimization 2}
 \max_{\theta \in \{0,1\}^N} \mathbb{E}w^T\theta + z_{1-\alpha} \sqrt{\theta^T (Q-D) \theta + \sum_{i} D_{ii} \theta_i}
\end{equation}
It can be seen that the objective of \eqref{eq: optimization 2} is concave, but equals the convex objective \eqref{eq: optimization} for all $\theta \in \{0,1\}^N$, since $x = x^2$ if $x$ is binary valued. \eqref{eq: optimization 2} can be solved or bounded by Gurobi using branch-and-bound methods. For faster computation, an upper bound can be found by relaxing the integrality constraint of \eqref{eq: optimization 2} to $\theta \in [0,1]^N$, resulting in a concave maximization problem whose global optimum can be efficiently found. Solving for $L$ given by \eqref{eq: L} proceeds analogously.



\section{Examples} \label{sec: examples}

\subsection{Social Networks and the Decision to Insure} \label{sec: cai}

In the experiment described in \cite{cai2015social}, rural farmers in China were randomly assigned to information sessions where they would be given the opportunity to purchase weather insurance.\footnote{data available at \texttt{https://www.aeaweb.org/articles?id=10.1257/app.20130442}} The farmers were instructed to list 5 close friends with whom they specifically discussed rice production or financial issues. First round sessions were held three days before second round sessions, so that first round attendees could have informal conversations with their second round friends, in which they might share their opinions about the insurance product. 


The goal of the experiment was to broadly demonstrate the importance of information sharing, by measuring its effects in a randomized setting. One of the conclusions of \cite{cai2015social} was that the decision to purchase insurance was affected not only by a farmer's own treatment assignment, but also by that of their friends; specifically, farmers assigned to a second round session were more likely to purchase insurance if more of their listed friends in the first round were assigned to a high-information session. 

Our analysis will allow that units may have been affected not only by the treatment of their friends, but other units as well. For example, units might be affected by the treatment of friends-of-friends, as information might travel further than ``1 hop'' on the friendship network. Additionally, if questions were allowed during the information sessions, then a question asked by one unit might affect the opinions of all other units in their session.



\paragraph{Setup}

We let $M$ and $N$ denote the number of units in the first and second rounds, and define $X$, $Y$, $Z$, and $\theta$ as follows. For the first round units, let $X \in \{0,1\}^M$ indicate whether each unit was assigned to a high or low information session. For the second round units, let $Y \in \{0,1\}^N$ indicate whether or not they purchased insurance; let $Z$ denote their number of first round friends assigned to a high information session; and let $\theta \in \{0,1\}^N$ denote their outcomes that would have occurred under a counterfactual in which the first round sessions were not held. For example, if $Y_i = 1$ and $\theta_i =0$, then unit $i$ purchased insurance, but would not have done so if the first round sessions were not held. Thus $A = Y - \theta$ equals the effect of the first round sessions on the purchasing decisions of the second round.


We will assume that $X$ is randomized by simple random sampling within each village, and let the classes $\{\Pi_k\}_{k=1}^K$ divide the second round units by village and number of first round friends. In fact, \cite{cai2015social} reports that treatments were stratified by family size and amount of land used for rice farming; however, further details are not included, and additionally these variables are missing for a significant fraction of the first round units. For this reason, our results for this example should be viewed primarily as a illustration of the proposed method.

\paragraph{Results}

Table \ref{table: cai results} gives estimates for various parameters described in Section \ref{sec: general case}. The point estimate for $\beta^\tadj$ suggests that for two farmers in the same propensity class, a unit difference in their number of high information first round friends was associated with a 0.08 difference in the effects of the first round sessions. However, the interval is weakly non-significant (95\% PI: [-0.01, 0.16]). 

For $\gamma$ defined by \eqref{eq: effect curve}, the contrasts $\gamma_z - \gamma_0$ suggest an effect threshold at 2 high information friends. Units with 1 high information friend experienced effects that were indistinguishable from those with zero such friends. In constrast, units with 2 high information friends were affected more positively than those with zero such friends, with 22 more purchases caused (or fewer purchases prevented) by the first round sessions per 100 second round individuals (95\% PI: [2, 41]). Contrasts $\gamma_z - \gamma_0$ for $z\geq 3$ were non-significant due to small numbers of units with $Z\geq 3$. Similar differences in effect were estimated by the weighted and matching-based estimands $\tau^\tweight$, $\tau^\tmatch$, and $E[\tau^\tmatch|X]$, using the thresholded exposures $W_i = 1\{Z_i \geq 2\}$. 

The interval widths are large, and allow that the effects may have been small. For example, plugging the lower bound of the PI for $\tau^\tmatch$ into \eqref{eq: w number of affected units} only implies that at least 6 (or 2\%) of the matched units might have been affected.

The interpretation of our estimates might be the following: under the observed treatment assignment, the second round units were affected by the first round sessions, with attributable effects that varied systematically with the information content given to their friends who attended first round sessions, and in particular with whether at least two such friends received high information content. This may be taken as evidence that information sharing affected the decisions of the farmers. However, without further assumptions, we cannot estimate whether the second round units were affected positively or negatively on average, nor whether the effects would be sustained if the treatment assignment was modified or redrawn.






\begin{table}[t!]
\begin{center}
\begin{tabular}{l|r |r | r | r }
   & Point Est. & Bias & 95\% PI \\
  \hline
  $\beta^\tadj$ & 8\% & $\pm$ 0.4\% & [-1\%, 17\%]  \\
  $\gamma_1 - \gamma_0$ & 2\% & $\pm$ 1\%& [-11\%, 14\%]  \\
  $\gamma_2 - \gamma_0$ & {\bf 21\%} & $\pm$ 1\%& [2\%, 40\%] \\
  $\gamma_3 - \gamma_0$ & 16\% & $\pm$ 2\% & [-17\%, 49\%]  \\
  $\gamma_4 - \gamma_0$ & 69\% &  $\pm$ 2\% & [-60\%, 195\%]  \\
  $\tau^\tweight$ & {\bf 19\%} & $\pm$ 2\%& [2\%, 36\%] \\
  $\tau^\tmatch$ & {\bf 20\%} & $\pm$ 1\% & [3\%, 37\%] \\
  $\operatorname{E}[\tau^\tmatch|X]$ & {\bf 19\%} & $\pm$ 1\% & [3\%, 35\%]  
\end{tabular}
\end{center}
\caption{Estimation Results for Insurance Experiment. (Bias denotes bounds given by \eqref{eq: bias bound 2})}
\label{table: cai results}
\end{table}

\subsection{Simulated Vaccine Trial} \label{sec: cholera sim}

The \texttt{vaccinesim} dataset simulates a vaccine trial using models that were partially fit to an actual cholera vaccine trial, as described in \cite{perez2014assessing}.\footnote{available in the R package \texttt{inferference}.} In both simulation and actual trial, units were divided into small contiguous neighborhoods ($\leq 20$ individuals). However, \cite{loh2018randomization} cite evidence of cross-neighborhood interference in the actual cholera trial, through shared bodies of water and kinship ties such as marriage.

 \paragraph{Setup} Let $N$ denote the number of individuals who participated in the trial. For $i \in [N]$, let $X_i$ denote whether unit $i$ received the vaccine or placebo; let $Y_i=1$ denote that the unit had cholera; let $\theta_i$ denote their outcome under the counterfactual in which all units had received the placebo; and let $V_i$ denote the fraction of $i$'s neighbors (including those who were non-participants) who were assigned to vaccination. Participants were assigned independently with probability 2/3 to receive the vaccine, otherwise receiving the placebo; non-participants had no opportunity to receive the vaccine. Supplement \ref{sec: data tables} contains additional simulation details.

\paragraph{Results} We estimate $\tau_1$ given by \eqref{eq: tau1}, and $\beta_1$, $\beta_2$, and $\beta_3$ in the following regression of effects against direct treatment, fraction of treated neighbors (including non-participants), an interaction term, and a control term,
\[ A_i \sim \beta_0 + \beta_1 X_i + \beta_2 V_i + \beta_3 (X_i\cdot V_i) + \beta_4 \, \mathbb{E}[V_i],\]
where the expectation of $V_i$ depends on the number of participants and non-participants in unit $i$'s neighborhood, and is used to control for the fact that $V_i$ is non-identically distributed.

Table \ref{table: vaccine sim estimates} shows results. The point estimate for $\tau_1$ suggests that vaccinated units were relatively more protected than placebo, with 12 additional cases prevented (or fewer cases caused) per 100 individuals under the observed treatment assignment. By plugging the lower bound of the prediction interval for $\tau_1$ into \eqref{eq: tau number of affected units}, we may say with 95\% confidence level that at least 101 participants (5.6\% of $N$) were affected by the vaccine trial. The point estimates for $\beta_1$ and $\beta_3$ suggest that the relative attributable effects between vaccine and placebo varied with neighborhood vaccination rates, with 
\[ 0.22 - 0.29 \cdot \text{(\% of treated neighbors)}\]
additional cases prevented per capita for vaccinated units compared to placebo. Thus in neighborhoods with low vaccination rates, the difference in attributable effects for vaccinated and placebo units may have been significantly larger than the overall difference given by $\tau_1$. The point estimate for $\beta_2$ indicates that placebo units were relatively more protected in neighborhoods with high vaccination rates; a difference of 1 additional vaccination per 100 individuals was associated with 0.6 additional cases prevented (or fewer cases caused) per 100 placebo-receiving individuals. Estimates for $\beta_1$, $\beta_2$ and $\beta_3$ had high uncertainty, and 90\% prediction intervals for these estimands are shown.

\begin{table}[t!]
\begin{center}
\begin{tabular}{l|l |l | l   l}
  & Point Est.  & Bias & PI & Coverage \\
  \hline
  $\tau_1$ & {\bf -0.12} & 0 & [-0.17, -0.07] & 95\%  \\
  $\beta_1$ &  {\bf -0.21} & $\pm 0.005$ & [-0.38, -0.05] & 90\%  \\
  $\beta_2$ &  {\bf -0.59} & $\pm 0.02$ &  [-1.16, -0.02] &  90\% \\
  $\beta_3$ &  0.29 & $\pm 0.01$ & [-0.03, 0.61] & 90\%
\end{tabular}
\end{center}
\caption{Estimates for \texttt{vaccinesim} data}
\label{table: vaccine sim estimates}
\end{table}

The practical interpretation of our point estimates might be the following. Under the observed treatment assignment, the relative attributable effects were consistent with those of a working vaccine that imbued herd resistance: units were relatively more protected when receiving the vaccine, and placebo units were relatively more protected when neighbors were vaccinated. As in the previous example, without further assumptions we cannot estimate the overall attributable effect of the vaccination trial, nor whether the effects would be sustained if the treatment assignment were modified or redrawn.

\section{Conclusions}

We have presented an approach for partial estimation of unit-level attributable effects in randomized experiments, with no assumptions regarding interference. For such settings, the approach goes beyond hypothesis testing, by answering whether the units were affected systematically as a function of their direct and indirect exposure to treatment, and by implying a lower bound on the number of affected units. 


\paragraph{Supplemental Materials} The supplement contains additional data examples; proof of \eqref{eq: tau number of affected units} and \eqref{eq: w number of affected units}; asymptotic results regarding consistency and interval coverage for regression estimands (including Proposition \ref{th: CLT tau}); and data summaries. 

\if1\blind
{
\paragraph{Acknowledgements} The author wishes to acknowledge helpful discussions with Brian Kovak, Dan Nagin, Beka Steorts, Eric Tchetgen Tchetgen, and Stefan Wager, as well as detailed and insightful feedback from the editors and reviewers
} \fi

\paragraph{Disclosures} The author reports there are no competing interests to declare.

 \bibliographystyle{apalike}
 \bibliography{bibfile}

 \pagebreak 
 
\appendix

\renewcommand{\theequation}{S.\arabic{equation}}
 \renewcommand{\thetable}{S\arabic{table}}%
 \setcounter{equation}{0}

\section*{Supplemental Material for ``Randomization-only Inference in Experiments with Interference''}

\section{Additional Data Analysis}

\paragraph{Coverage of Confidence Intervals in Simulated Vaccine Trial}  Using the simulation parameters given in \cite{perez2014assessing}, we may generate the unknown counterfactual outcomes $\theta$ and simulate $X$ to measure the coverage rates of our confidence intervals. For the model-generated values of $\theta$, in 10K simulations the interval for $\tau_1$ was well-calibrated, while the intervals for $\beta_1$, $\beta_2$, and $\beta_3$ achieved coverage but were as much as 2.2 times wider than necessary. However, for other values of $\theta$, found by solving \eqref{eq: U}, the coverage of the 90\% interval was as low as 91\%. As no assumptions are placed on $\theta$, this suggests that the interval widths were necessary in order to have coverage for all possible values of $\theta$.

\paragraph{Real Vaccine Trial (Aggregate Outcomes)} \label{sec: cholera real}

The paper \citeappendix{ali2005herd} analyzes the actual results of an actual cholera vaccine trial, in order to investigate the possibility of indirect (herd) protection under high levels of vaccination. A table of aggregate outcomes for the actual cholera vaccine trial is given, which we aggregate further in Table \ref{table: vaccine real data} into two groups of roughly equal size, by dividing the units according to whether their neighborhood vaccination rate was  $\leq 40\%$ or $> 40\%$. We consider the regression 
\[ A_i \sim \beta_1 X_i + \beta_2 (X_i \cdot i \in \textup{group 1}) + \textup{controls},\]
where the control terms are the membership indicators for groups 1 and 2. 

If no further assumption are placed on $\theta$,then we allow for the possibility that in the absence of the vaccinations, the cholera rate could have been hundreds of times larger than the observed rate of 2.8 cases per thousand. As this may be unrealistic, we will consider an additional assumption based on past cholera rates. The trial covered the Matlab region of Bangladesh in 1985, and cholera rates for this region are given in \citeappendix[Fig. 1]{zaman2020can}. We will assume that under the counterfactual $\theta$ in which all units receive the placebo, the cholera rate for the participants would be upper bounded by 7 cases per thousand, 
\begin{equation}\label{eq: theta assumption 1}
 \frac{1}{N} \sum_{i=1}^N \theta_i \leq 0.007,
 \end{equation}
which is substantially higher than the highest observed rate of 5 per thousand in the ten years before the trial. As the participants ($\approx 30\%$ of the population) are not a random sample, some caution is required; however, \eqref{eq: theta assumption 1} may be easier to consider than a bound on interference between neighborhoods. 

The coefficients $\beta_1$ and $\beta_2$ characterize the relative effects of receiving the vaccine versus placebo, in each group. Table \ref{table: vaccine real data estimates} shows the results. The estimates suggest that vaccinated units were relatively more protected that placebo ones, by a difference of
\[ 1.3 + 2.2 \cdot 1\{\text{unit is in group 1}\} \]
additional cases prevented per thousand. Prediction intervals (weakly) exclude zero for the interaction term $\beta_2$, indicating that the relative effects were different for the two groups, and also for $\beta_1 + \beta_2$, which gives the relative effects for group 1. These results suggest that the overall difference in outcomes between treatment and control, which equaled 2.4 cases per thousand, may under-represent the effects of the vaccine, as larger relative effects occurred when vaccination rates were low. As in previous examples, without further assumptions we cannot estimate the overall attributable effect of the vaccination trial, nor whether the effects would be sustained if the treatment assignment were modified or redrawn.

\begin{table}[t!]
\begin{center}
\begin{tabular}{l|l|l |l   }
   & Vaccination & \multicolumn{2}{c}{Cholera Rate (per 1000)} \\
  Group & Rate & Treated & Control \\
  \hline
  1 & $\leq 40\%$ & 2.1 (54/25K) & 5.6 (72/13K) \\  
  2 & $> 40\%$ & 1.7 (42/25K) & 3.0 (36/12K) \\
  \hline
  Total &  &  1.94 (96/49K) & 4.4 (108/25K)
\end{tabular}
\end{center}
\caption{Cholera cases in treated and control units grouped by neighborhood vaccination rates. (Taken from \cite{hudgens2008toward})}
\label{table: vaccine real data}
\end{table}

\begin{table}[t!]
\begin{center}
\begin{tabular}{l|l |l | l  }
  & Point Estimate  & Bias & 90\% PI  \\
  \hline
  $\beta_1$ & -1.3 & 0 & [-3.5, 0.9] \\
  $\beta_2$ & -2.2 & 0 & [-4.4, 0.0] \\
  $\beta_1 + \beta_2$ & -3.5 & 0 & [-5.6,\, -1.4]
\end{tabular}
\end{center}
\caption{Estimates (in cases per thousand) for aggregate-level cholera data}
\label{table: vaccine real data estimates}
\end{table}

%

\paragraph{Support for Conflict Reduction in Schools} \label{sec: aronow}

The paper \citeappendix{paluck2016changing} describes an experiment\footnote{data available at \texttt{https://www.icpsr.umich.edu/web/civicleads/studies/37070}} in which public middle schools were randomized to receive a treatment designed to encourage anti-conflict norms and behavior in students. Within each treated schools, a small subset of eligible students (between 40-64) were stratified into 4 groups and then randomized between treatment and control. Treated students were invited to participate in bi-monthly meetings in which they designed and enacted strategies to reduce schoolwide conflict. Outcomes for all students in each school included the self-reported wearing of an orange wristband representing support for anti-conflict norms.  Network information (``spent time with'') was recorded for all students. Interference seems likely to extend beyond 1 hop on this network; as each invited unit could potentially influence school-wide strategy by participating in the bi-monthly meetings, the treatment of a single unit might affect the outcomes of their entire school.  Estimation of effects on the wristband outcome has been used in \cite{aronow2012general} and \cite{leung2019causal} to illustrate their proposed methods, as we do here.

\begin{table}[t!]
\begin{center}
\begin{tabular}{l|l |l | l  }
    Definition of $W_i=0$ & Point Estimate  & Bias & 95\% PI  \\
  \hline
  $i$ is treated, with no treated friend & -0.8 & $\pm 0.04$ & [-0.28, 0.12]\\
  $i$ is control, with treated friend & -0.13 & $\pm 0.03$ & [-0.27, 0.00] \\
  $i$ is control, with no treated friend & -0.24 & $\pm 0.03$ & [-0.48, 0.00] 
\end{tabular}
\end{center}
\caption{Estimates of $\mathbb{E}[\tau^{\tmatch}]$, where $W_i=1$ if unit $i$ is treated and has a treated friend}
\label{table: aronow}
\end{table}

Our setup is as follows: Let $X_i$ denote if student $i$ was invited to participate in the anti-conflict meetings (which can only equal 1 for eligible students), and $Y_i$ their outcome of whether they self-reported wearing the wristband representing support for anti-conflict norms. Breaking from the previous examples, we define the counterfactual $\theta_i$ to denote their outcome if all eligible units were invited to participate in the meetings (``full treatment''). With $\theta$ so defined, the difference $Y - \theta$ is interpretable as the effect of not inviting the control units to participate. 

To estimate the relative effects of direct treatment, we estimated $\mathbb{E}[\tau^\tmatch]$, where the matching was constructed between the treated and control units in each school. This resulted in a point estimate of 0.15, implying that not inviting the control units to participate in the meetings affected the control units more adversely than the treated units, with 15 fewer cases of self-reported wristband wearing caused by the meetings per 100 units (95\% PI: [7, 24] per 100). 

The units were then divided into 4 groups according to their direct treatment and whether or not they had at least one treated friend. Between each pair of groups a matched sample was constructed, where matched units had the same school and number of eligible friends. Table \ref{table: aronow}  shows estimates of $\mathbb{E}[\tau^\tmatch]$ for a subset of pairings. Similar results were found using weighting or matching. Weakly non-significant differences in effects was found between some pairs of groups; as a result of not inviting the control units to participate, units in control (regardless of whether they had a treated friend) were affected more adversely than those who were treated and had at least one treated friend.

As in previous examples, without further assumptions we cannot estimate the overall attributable effect of the experiment, nor whether the observed effects would be sustained if the treatment assignment were modified or redrawn.

\paragraph{Facebook voting experiment} \label{sec: facebook}

The unpublished manuscript \cite[Sec 6.3]{choi2018manuscript} estimates $\tau_1$ for a voting experiment conducted by Facebook during the 2010 US congressional elections, originally published in \citeappendix{bond201261}. Their results imply that treated units were more strongly encouraged than control units to report that they had voted, with 2.2 additional reports caused by the experiment per 100 individuals (95\% PI: [2.1, 2.3]). As in previous examples, without further assumptions we cannot estimate the overall attributable effect of the experiment, nor whether the effects would be sustained if the treatment assignment were modified or redrawn.

\section{Asymptotics} \label{sec: asymptotics}

\subsection{Consistency of Regression Estimates} To study consistency of regression estimates, we will assume that the treatment vector $X$ can be written as a function $X(U)$ of a vector $U = (U_1,\ldots,U_N)$ consisting of i.i.d.\! uniform random variables,
\begin{align}
  \label{eq: U design 1} U_i & \stackrel{\textup{iid}}{\sim} \textup{Unif}\, [0,1], &  i \in  [N]  
\end{align}
where $X(U)$ satisfies the following bounded differences condition: if the vectors $U$ and $U'$ differ in only a single entry, then 
\begin{equation} \label{eq: CX}
\sum_{i=1}^N 1\{X_i(U) \neq X_i(U')\} \leq C,  
\end{equation}
for some constant $C$. This accomodates a variety of experiment designs. For example, if $X$ is given by
\begin{align} \label{eq: bernoulli U}
 X_i = \begin{cases} 1 & \textup{ if $U_i \leq p$} \\ 0 & \textup{ else,} \end{cases}
 \end{align}
then each unit is independently assigned to treatment with probability $p$ and \eqref{eq: CX} holds with $C=1$. To assign treatment by sampling $M$ units without replacement, we may define $X$ by
\begin{align} \label{eq: SWOR U}
 X_i = \begin{cases} 1 & \textup{ if $U_i \leq U_{[M]}$} \\ 0 & \textup{ else,} \end{cases}
\end{align}
where $U_{[1]} \leq \cdots \leq U_{[N]}$ are the sorted values of $U_1,\ldots,U_N$, so that the units with the $M$ smallest $U_i$ values are assigned to treatment. Then \eqref{eq: CX} holds for $C=2$. (To extend to stratified sampling, $U_1,\ldots,U_N$ can be sorted within each strata, in which case $C=2$ as well.)

Proposition \ref{th: regression consistency} considers regression estimands $(\beta, \gamma)$ of the form
\begin{equation}\label{eq: regression consistency}
 (\beta, \gamma) = \argmin_{\beta \in \mathbb{R}^d, \gamma \in \mathbb{R}^K} \sum_{i=1}^N \left(A_i - \xi_i(X)^T \beta - \phi_i^T \gamma\right)^2,
 \end{equation}
which includes \eqref{eq: general regression beta}, and gives conditions under which point estimates will be consistent at rate $O_P(N^{-1/2})$. It assumes a sequence of experiments for which the following assumptions hold. Assumption \ref{as: X U} pertains to the experiment design, and requires the treatment assignment $X$ to have limited dependence between unit assignments, as enforced by the bounded differences condition \eqref{eq: CX}. 

\begin{assumption} \label{as: X U}
The treatments $X \equiv X(U)$ satisfies the bounded differences condition \eqref{eq: CX} 
\end{assumption}

Assumption \ref{as: bounded differences xi} requires each basis function $\xi_i(X)$ to depend only on the treatment of unit $i$ and $i$'s neighbors in a network $G$ that has bounded degree.

\begin{assumption} \label{as: bounded differences xi}
For some network $G$ with in-degree and out-degree bounded by $d_{\max}$, each mapping $\xi_i$ satisfies 
\[ \xi_i(X) = \xi_i(X') \textup{ for all } X, X': X_i = X_i' \textup{ and } X_j = X_j' \textup{ for all } j: G_{ij} = 1\]
so that $\xi_i(X)$ depends only on the treatment of unit $i$ and $i$'s neighbors in $G$.
\end{assumption}

Assumption \ref{as: bounded vi} states that $\xi_i(X)$ and $\phi_i$ are bounded in magnitude. 

\begin{assumption} \label{as: bounded vi}
It holds that
\[ \max_i \|\phi_i\|_\infty \leq B \quad \text{ and } \quad \max_i \|\xi_i(X)\|_\infty \leq B \text{ with probability 1} \]
\end{assumption}

Assumption \ref{as: lambda min} requires the expectation of the regressors to be bounded away from collinearity, in that their expected gram matrix should have smallest eigenvalue bounded away from zero. 

\begin{assumption} \label{as: lambda min}
Let $v_i = \left[ \begin{array}{c} \xi_i(X) \\ \phi_i \end{array}\right]$ for $i \in [N]$. The matrix $\bar{M} = \frac{1}{N}\sum_{i=1}^N \mathbb{E}[v_i v_i^T]$ has smallest eigenvalue bounded below by $\lambda_{\min}$
\end{assumption}

Assumption \ref{as: phi} requires the regression specification to correctly adjust for non-identically distributed regressors $\{\xi_i(X)\}_{i=1}^N$, enabling unbiased estimation of $\beta$.  It gives two possible approaches for doing so. The first is that the expectation of $\xi_i(X)$ is used as a control variable. The second is that the units can be divided into propensity classes, where two units belong to the same class if their regressors have the same expected value.

\begin{assumption} \label{as: phi}
The regressors $(\xi_i, \phi_i)_{i=1}^N$ satisfy either of the following two conditions:
\begin{enumerate}[label=A.\arabic*,ref=A.\arabic*]  \addtocounter{enumi}{0} 
  \item \label{as: phi 1} $\phi_i = \mathbb{E}\xi_i(X)$ for $i \in [N]$
  \item \label{as: phi 2} There exists a mapping $\kappa:[N] \rightarrow [K]$ that partitions the $N$ units into $K$ classes, and vectors $\bar{\xi}_1,\ldots,\bar{\xi}_K$ such that for all $i \in [N]$, it holds that 
  \[ \mathbb{E}\xi_i(X) = \bar{\xi}_{\kappa(i)}  \qquad \phi_i = e_{\kappa(i)}, \qquad i \in [N]\]
  where $e_k$ is the $k$th column of the $K \times K$ identity matrix
\end{enumerate}
\end{assumption}

The statement of the proposition is as follows.

\begin{proposition} \label{th: regression consistency}

Assume a sequence of experiments with estimands $(\beta, \gamma)$ given by \eqref{eq: regression consistency} and point estimates $(\hat{\beta}, \hat{\gamma})$ given by
\begin{equation}\label{eq: regression consistency estimator}
 (\hat{\beta}, \hat{\gamma}) = \argmin_{\beta \in \mathbb{R}^d, \gamma \in \mathbb{R}^K} \sum_{i=1}^N \left(Y_i - \xi_i(X)^T \beta - \phi_i^T \gamma \right)^2,
 \end{equation}
where $d$ and $K$ are constant as $N \rightarrow \infty$. Let assumptions \ref{as: X U} - \ref{as: phi} hold, where $C$, $d_{\max},$ and $\lambda_{\min}$ are constant as $N \rightarrow \infty$. Then it holds that $\hat{\beta} - \beta = O_P(N^{-1/2})$.
\end{proposition}

\subsection{Coverage of prediction intervals}

Proposition \ref{th: regression normality} gives conditions under which prediction intervals for regression-based estimands will asymptotically have correct coverage. It assumes a sequence of experiments, and applies a central limit theorem for sums of variables whose dependency graph has bounded degree \citepappendix[Thm 2.2]{rinott1994normal}. 

To apply the central limit theorem of \citepappendix[Thm 2.2]{rinott1994normal}, we require Assumptions \ref{as: bounded differences xi}-\ref{as: lambda min}, and additionally that the treatments are iid Bernoulli randomized:
\begin{assumption}\label{as: bernoulli}
$X$ has distribution
\[ X_i \stackrel{iid}{\sim} \operatorname{Bernoulli}(\rho)\]
for $0 < \rho < 1.$
\end{assumption}
\noindent As mentioned in the discussion of \eqref{eq: CX}, Assumption \ref{as: bernoulli} implies that Assumption \ref{as: X U} holds with $C=1$.





The statement of Proposition \ref{th: regression normality} is as follows. 


\begin{proposition} \label{th: regression normality}
Assume a sequence of experiments with estimands $(\beta,\gamma)$ given by \eqref{eq: regression consistency} and point estimates $(\hat{\beta}, \hat{\gamma})$ given by \eqref{eq: regression consistency estimator}, where $d$ and $K$ are constant as $N \rightarrow \infty$. 

For $i \in [N]$, let $v_i = \left[ \begin{array}{c} \xi_i(X) \\ \phi_i \end{array}\right]$. For any $\ell \in [d]$, let $w$, $\bar{w}$, and $\ddot{w} \in \mathbb{R}^N$ be given by
\begin{align}
\label{eq: w} w_i & = e_\ell^T\left(\sum_{i=1}^N v_i v_i^T\right)^{-1}  v_i \\
\label{eq: bar w} \bar{w}_i & = e_\ell^T \left(\sum_{i=1}^N \mathbb{E}v_i v_i^T\right)^{-1} \mathbb{E}v_i\\
\label{eq: ddot w} \ddot{w}_i & =  e_\ell^T \left(\sum_{j=1}^N \mathbb{E}v_j v_j^T\right)^{-1}  \left( v_i - \sum_{j=1}^N v_j v_j^T \left(\sum_{k=1}^N \mathbb{E}v_k v_k^T\right)^{-1} \mathbb{E}v_i\right)
\end{align}
and let $U$ and $L$ respectively equal
\begin{align}
\label{eq: regression normality U} U &= \max_{\theta \in \{0,1\}^N} \left[ \bar{w}^T \theta + z_{1-\alpha} \sqrt{\operatorname{Var}(\ddot{w}^T \theta)}\right] \\
\label{eq: regression normality L} L &= \min_{\theta \in \{0,1\}^N} \left[ \bar{w}^T \theta - z_{1-\alpha} \sqrt{\operatorname{Var}(\ddot{w}^T \theta)}\right],
\end{align}
and for any $\epsilon > 0$ let $U'$ and $L'$ equal
\begin{align}
\label{eq: U'} U' &= \max_{\theta \in \{0,1\}^N} \left[\bar{w}^T\theta + \frac{\epsilon}{N^{5/9}}\right] \\
\label{eq: L'} L' & = \min_{\theta \in \{0,1\}^N} \left[\bar{w}^T\theta - \frac{\epsilon}{N^{5/9}}\right].
\end{align}

Let Assumptions \ref{as: bounded differences xi}-\ref{as: lambda min} and \ref{as: bernoulli} hold, where $C$, $d_{\max}$, $\lambda_{\min}$, $\rho$, $\sigma_{\min}^2$, and $\epsilon$ are constant as $N \rightarrow \infty$. Then $\beta_\ell$ is contained by the interval $\hat{\beta}_\ell - [\max(U, U'), \, \min(L, L')]$ with probability asymptotically lower bounded by $(1-2\alpha)$.
\end{proposition}

\paragraph{Discussion} For technical reasons, Proposition \ref{th: regression normality} uses interval boundaries that differ slightly from those given in \eqref{eq: U}-\eqref{eq: L}, in two respects:
\begin{enumerate}
  \item The terms $\mathbb{E}w^T\theta$ and $\operatorname{Var}(w^T\theta)$ appearing in \eqref{eq: U}-\eqref{eq: L} have been replaced by $\bar{w}^T\theta$ and $\operatorname{Var}(\ddot{w}^T\theta)$ in \eqref{eq: regression normality U}-\eqref{eq: regression normality L}, where $\bar{w}$ and $\ddot{w}$ are given by \eqref{eq: bar w}-\eqref{eq: ddot w}. This arises because our proof applies Taylor approximation to a function $h_\theta(T)$ representing $(\hat{\beta}_\ell - \beta_\ell)$, the estimation error of the regression estimate. This centers the asymptotic distribution at $h_\theta(\mathbb{E}T)$ rather than at $\mathbb{E}h_\theta(T)$, and shows weak convergence of $h_\theta(T) - h_\theta(\mathbb{E}T)$ to its linearized approximation $\nabla h_\theta(\mathbb{E}T)^T(T - \mathbb{E}T)$, which will be proven to equal $\ddot{w}^T\theta$, and hence to have variance $\operatorname{Var}(\ddot{w}^T\theta)$. 

  In our data examples, using $(L,U)$ defined by either \eqref{eq: regression normality U}-\eqref{eq: regression normality L} or \eqref{eq: U}-\eqref{eq: L} gave identical results, up to the precision shown in our tables. If we assume bounded third moments for the estimation error normalized by the standard deviation of its linearized approximation,
  \[ \sup_N \ \mathbb{E}\left[ \left| \frac{\hat{\beta}_\ell - \beta_\ell}{\sqrt{\operatorname{Var}\left[\nabla h(\mathbb{E}T)^T(T -\mathbb{E}T)\right]}} \right|^3 \right] < \infty, \]
  then Proposition \ref{th: regression normality} holds with $\mathbb{E}w^T\theta$ and $\operatorname{Var(w^T\theta)}$ used in place of $\bar{w}^T\theta$ and $\operatorname{Var}(\ddot{w}^T\theta)$. 

  \item We take the union $[L,\, U]$ with a ``backup'' interval whose variance component is $o(N^{-1/2})$. This ensures coverage even if the variance of the estimation error is so small that asymptotic normality does not hold, due to degeneracy. 
\end{enumerate}

\section{Proofs}  \label{sec: proofs}

\subsection{Proof of \eqref{eq: tau number of affected units} and \eqref{eq: w number of affected units}}

By H\"older's inequality, it holds that for any $w, A \in \mathbb{R}^N$ that
\[ \left| w^TA \right| \leq \left(\max_{i} |w_i|\right) \cdot \left(\sum_{i=1}^N |A_i|\right) \]
Since $A = Y - \theta$, where $Y$ and $\theta$ are binary valued, we may substitue $|A_i| = 1\{Y_i \neq \theta_i\}$ and rearrange to show
\[ \sum_{i=1}^N 1\{Y_i \neq \theta_i\} \geq \frac{\left| w^T A \right|}{ \max_{i} |w_i| } \]
It can be seen that $\tau_1 = w^T A$ for $w$ given by
\begin{equation} \label{eq: tau1 as w}
 w_i = \begin{cases} 1/N_1 & \text{ if } X_i=1 \\ -1/N_0 & \text{ if } X_i = 0 \end{cases},
\end{equation}
Substituting this and \eqref{eq: tau1 as w} yields
\[ \sum_{i=1}^N 1\{Y_i \neq \theta_i\} \geq |\tau_1| \cdot \max( N_1, N_0 ) \]
proving \eqref{eq: tau number of affected units}.

To show \eqref{eq: w number of affected units}, we observe that H\"older's inequality also implies that
\[ \left| \sum_{i: w_i \neq 0} w_i A_i \right| \leq \max_{i: w_i \neq 0} |w_i| \cdot \sum_{i: w_i \neq 0} |A_i| \]
Substituting $\sum_{i: w_i \neq 0} w_i A_i = w^T A$ and proceeding similarly proves \eqref{eq: w number of affected units}.


\subsection{Proof of Proposition \ref{th: CLT tau}} \label{sec: CLT tau proof}

Our proof of Propositions \ref{th: CLT tau} will use the following finite population central limit theorem.

\begin{theorem}{\citepappendix[Thm. 1]{li2017general}} \label{th: finite pop CLT}
Let $\bar{u}_S$ denote the average of $N_1$ units sampled without replacement from the finite population $u_1,\ldots,u_N$, with $\bar{u} = N^{-1} \sum_i u_i$.  As $N \rightarrow \infty$, if
\begin{equation} \label{eq: finite pop CLT condition}
 \frac{1}{\min(N_1,N-N_1)}\cdot\frac{m_N}{v_N} \rightarrow 0,
 \end{equation}
where $m_N$ and $v_N$ are given by
\[ m_N = \max_i  (u_i - \bar{u})^2 \qquad \text{and} \qquad v_N = \frac{1}{N-1} \sum_i (u_i - \bar{u})^2,\]
then $(\bar{u}_S - \bar{u}) / \sqrt{\operatorname{Var}(\bar{u}_S)} \rightarrow N(0,1)$ in distribution.
\end{theorem}

Lemma \ref{le: diff means var} gives a simpler expression for the difference in means of $\theta$, from which its variance can be derived.

\begin{lemma} \label{le: diff means var}
  Let $X$ be assigned by sampling $N_1$ units without replacement, and let $\bar{\theta}$ and $\sigma_{\theta}^2$ denote the mean and variance of $\{\theta_i\}_{i=1}^N$. Then it holds that 
  \[ \sum_{X_i=1} \frac{\theta_i}{N_1}  - \sum_{X_i=0} \frac{\theta_i}{N_0} = \frac{1}{N_1} \sum_{i=1}^N X_i \frac{(\theta_i - \bar{\theta})N}{N_0}, \]
  which is the average of $N_1$ samples without replacement from the finite population $u_1,\ldots,u_N$, where $u_i = (\theta_i - \bar{\theta})N/N_0$, and which has variance
    \begin{align} \label{eq: proof CLT tau var}
      \frac{1}{N_1}\frac{N_0}{N-1}\frac{N^2}{N_0^2} \sigma_\theta^2.
    \end{align}
\end{lemma}
  
\begin{proof}[Proof of Lemma \ref{le: diff means var}]
  \begin{align*}
    & \sum_{X_i=1} \frac{\theta_i - \bar{\theta}}{N_1}  - \sum_{X_i=0} \frac{\theta_i - \bar{\theta}}{N_0} \\
    & \hskip.3cm = \sum_{i=1}^N \frac{X_i(\theta_i - \bar{\theta})}{N_1}  - \sum_{i=1}^N \frac{(1-X_i)(\theta_i - \bar{\theta})}{N_0} \\
    & \hskip.3cm  = \sum_{i=1}^N X_i (\theta_i - \bar{\theta}) \frac{N}{N_1N_0}  - \frac{1}{N_0} \underbrace{\sum_{i=1}^N (\theta_i - \bar{\theta})}_{=0} \\    
    & \hskip.3cm  = \frac{1}{N_1} \sum_{i=1}^N X_i \frac{(\theta_i - \bar{\theta})N}{N_0}.
  \end{align*}
  This equals the average of $N_1$ samples without replacement from the finite population $u_1,\ldots,u_N$, where $u_i = (\theta_i - \bar{\theta})N/N_0$, which is known to have variance equal to $\frac{1}{N_1}\frac{N_0}{N-1}$ multiplied by the variance of $\{u_i\}$, which equals \eqref{eq: proof CLT tau var}.    
\end{proof} 
\begin{proof}[Proof of Proposition \ref{th: CLT tau}]
  It can be seen that $\tau_1 - \hat{\tau}_1$ equals
  \[\tau_1 - \hat{\tau}_1 = \sum_{X_i=1} \frac{\theta_i}{N_1}  - \sum_{X_i=0} \frac{\theta_i}{N_0}. \]
By Lemma \ref{le: diff means var}, this equals the average of $N_1$ samples without replacement from the finite population $u_1,\ldots,u_N$, where $u_i = (\theta_i - \bar{\theta})N/N_0$, with variance given by \eqref{eq: proof CLT tau var}. Hence the terms $m_N$ and $v_N$ as defined in Theorem \ref{th: finite pop CLT} can be bounded by
    \begin{equation}\label{eq: proof CLT tau constants}
    m_N \leq (N/N_0)^2 \qquad \text{and} \qquad v_N = \sigma_\theta^2 (N/N_0)^2.
    \end{equation}
  
  To prove the proposition, we will divide the sequence of experiments into two subsequences: (i) those for which $\sigma_\theta^2 \geq (\log N)/N$, and (ii) those for which $\sigma_\theta^2 < (\log N)/N$, and then show coverage of the proposed interval for each subsequence separately.
  
  For subsequence (i), we observe that by \eqref{eq: proof CLT tau constants}, along with $N_1 \rightarrow \rho N$ and $\sigma_\theta^2 \geq (\log N)/N$, it can be seen that the CLT condition \eqref{eq: finite pop CLT condition} holds, i.e., 
  \begin{equation*} 
   \frac{1}{\min(N_1,N-N_1)}\cdot\frac{m_N}{v_N} \rightarrow 0.
   \end{equation*}    
  Thus by Theorem \ref{th: finite pop CLT}, it holds that
  \begin{equation}\label{eq: proof CLT tau CLT}
   \mathbb{P}\left( |\tau_1 - \hat{\tau}_1| \geq z_{1-\alpha} \sqrt{\operatorname{Var}(\tau_1 - \hat{\tau}_1)} \right) \rightarrow 1 - 2\alpha.
  \end{equation}
  Since $\theta \in \{0,1\}^N$, it holds that $\sigma_\theta^2 \leq 1/4$, and plugging this bound into \eqref{eq: proof CLT tau var} yields an upper bound on the variance,
  \begin{align} \label{eq: proof CLT tau var bound}
    \operatorname{Var}(\tau_1 - \hat{\tau}_1) & \leq \frac{1}{4}\frac{N}{N-1}\frac{N}{N_1 N_0},
  \end{align}
  and combining \eqref{eq: proof CLT tau var bound} and \eqref{eq: proof CLT tau CLT} implies for subsequence (i) that
  \begin{equation}\label{eq: proof CLT tau CLT 2}
   \mathbb{P}\left( |\tau_1 - \hat{\tau}_1| \geq z_{1-\alpha} \sqrt{\frac{1}{4}\frac{N}{N-1}\frac{N}{N_1 N_0}} \right) \rightarrow  2\alpha,
  \end{equation}
  and hence that $\tau_1$ is in the interval $\hat{\tau}_1 \pm z_{1-\alpha} \sqrt{\frac{1}{4}\frac{N}{N-1}\frac{N}{N_1 N_0}}$ with the same probability.
  
  For subsequence (ii), plugging $\sigma_\theta^2 < (\log N)/N$ into \eqref{eq: proof CLT tau var} implies that $\sqrt{\operatorname{Var}(\tau_1 - \hat{\tau}_1)}$ is a vanishing fraction of $z_{1-\alpha} \sqrt{\frac{1}{4}\frac{N}{N-1}\frac{N}{N_1 N_0}}$, the width of the proposed confidence interval. As a result, for this subsequence it holds by Chebychev's inequality that 
  \begin{equation}\label{eq: proof CLT tau CLT 3}
   \mathbb{P}\left( |\tau_1 - \hat{\tau}_1| \geq z_{1-\alpha} \sqrt{\frac{1}{4}\frac{N}{N-1}\frac{N}{N_1 N_0}} \right) \rightarrow 0, 
  \end{equation}
  and hence that $\tau_1$ is in the interval $\hat{\tau}_1 \pm z_{1-\alpha} \sqrt{\frac{1}{4}\frac{N}{N-1}\frac{N}{N_1 N_0}}$ with the same probability.
  
  Combining \eqref{eq: proof CLT tau CLT 2} for subsequence (i) and \eqref{eq: proof CLT tau CLT 3} for subsequence (ii) implies that for combined sequence of experiments, $\tau_1$ is in the interval $\hat{\tau}_1 \pm z_{1-\alpha} \sqrt{\frac{1}{4}\frac{N}{N-1}\frac{N}{N_1 N_0}}$ with probability asymptotically lower bounded by $1- 2\alpha$, proving the proposition.
  
%
%
  
\end{proof}

\subsection{Proof of Proposition \ref{th: regression consistency}}

The proof of Proposition \ref{th: regression consistency} will use the Azuma-Hoeffding (or McDiarmid's) inequality \citepappendix[Thm 6.2]{boucheron2013concentration}, which states that given independent variables $U = (U_1,\ldots,U_N)$, and a function satisfying the bounded difference property
\[ |f(U) - f(U')| \leq c_i \text{ if } U_j = U_j' \text{ for all } j \neq i,\]
it holds that $\mathbb{P}\left( | f(U) - \mathbb{E}f | > t\right) \leq 2\exp\left( - 2t^2/\sum_i c_i^2 \right)$.

Additionally, we will use a perturbation bound on matrix inverse found in \citeappendix{stewart1969continuity}
\begin{lemma}\citeappendix[Eq. 1.5]{stewart1969continuity} \label{le: matrix inverse bound}
If a matrix $M$ has inverse $M^{-1}$, then for any matrix norm $\| \cdot\|$ with $\|I\|=1$, if it holds that
\[ \|M^{-1}\|\, \|E\| < 1,\]
then
\begin{equation}\label{eq: matrix inverse bound}
\frac{\|(M + E)^{-1} - M^{-1}\|}{\|M^{-1}\|} \leq \frac{\kappa(M) \|E\|/\|M\|}{1 - \kappa(M) \|E\|/\|M\|},
\end{equation}
where $\kappa(M) = \|M\|\, \|M^{-1}\|$.
\end{lemma}
As a consequence of Lemma \ref{le: matrix inverse bound}, given positive definite $M$ with smallest eigenvalue lower bounded by $\lambda_{\min}$ and operator norm $\|M\|_{\textup{op}} < B$, and $\|E\|_{\textup{op}} \leq \min(\lambda_{\min}, \lambda_{\min}^2/B)$, then
\begin{align}
 \label{eq: my inverse bound}\|(M+E)^{-1} - M^{-1}\|_\textup{op} &\leq \frac{\frac{B}{\lambda_{\min}^3} \|E\|_{\textup{op}}}{1 - \frac{B}{\lambda_{\min}^2}\|E\|_{\textup{op}}}, 
\end{align}
implying uniform continuity of matrix inverse over positive definite matrices with smallest eigenvalue at least $\lambda_{\min}.$

\begin{proof}[Proof of Proposition \ref{th: regression consistency}]
By the definitions of $(\beta,\gamma)$ in \eqref{eq: regression consistency} and of $(\hat{\beta}, \hat{\gamma})$ in \eqref{eq: regression consistency estimator}, linearity of least-squares estimation, and $A_i = Y_i - \theta_i$, it can be seen that
\begin{align*}
 \left[\begin{array}{c} \hat{\beta} - \beta \\ \hat{\gamma} - \gamma \end{array}\right] & = \left(\frac{1}{N} \sum_{i=1}^N v_i v_i^T \right)^{-1} \cdot \frac{1}{N} \sum_{i=1}^N (Y_i - A_i) v_i \\
 & = \left(\frac{1}{N} \sum_{i=1}^N v_i v_i^T \right)^{-1} \cdot \frac{1}{N} \sum_{i=1}^N \theta_i v_i \\
 & = M^{-1} b,
 \end{align*}
where $M$ and $b$ are given by
\[ M = \frac{1}{N}\sum_{i=1}^N v_i v_i^T \qquad b = \frac{1}{N} \sum_{i=1}^N \theta_i v_i.\]
Writing $M(U)$ and $b(U)$ to make their dependence on $U$ explicit, it can be seen that if $U$ and $U'$ differ only at a single entry, then
\[ |M_{ij}(U) - M_{ij}(U')| \leq \frac{C d_{\max} B^2}{N}  \qquad  | b_i(U) - b_i(U') | \leq \frac{C d_{\max} B}{N}, \qquad i,j \in [N],\]
as changing $U$ in a single entry changes at most $C$ entries of $X$, which changes at most $Cd_{\max}$ members of $\{\xi_i(X)\}_{i=1}^N$ and hence of $\{v_i\}_{i=1}^N$. As $\| v_i\|_\infty \leq B$ by Assumption \ref{as: bounded vi}, by Azuma-Hoeffding inequality it holds that 
\begin{align*}
 \mathbb{P}\left(| M_{ij} - \mathbb{E}M_{ij}| > \frac{t}{\sqrt{N}} \right) & \leq 2 \operatorname{exp}\left(-\frac{2t^2}{(C d_{\max} B^2)^2}\right) \\
  \mathbb{P}\left(|b_{i} - \mathbb{E}b_{i}| > \frac{t}{\sqrt{N}} \right) & \leq 2 \operatorname{exp}\left(-\frac{2t^2}{(C d_{\max} B)^2}\right) 
 \end{align*}
It follows that 
\begin{align*}
\mathbb{E}M_{ij} - M_{ij} &= O_P(N^{-1/2}), & \mathbb{E}b_i - b_i & = O_P(N^{-1/2}), & i,j \in [N] 
\end{align*}
Letting $\bar{M}$ and $\bar{b}$ abbreviate $\mathbb{E}M$ and $\mathbb{E}b$, it follows that 
\begin{align*}
 \left[\begin{array}{c} \hat{\beta} - \beta \\ \hat{\gamma} - \gamma \end{array}\right] & = M^{-1} b \\ 
 & = \left[ \bar{M} + O_P\left(N^{-1/2}\right)\right]^{-1} \cdot \left(\bar{b} + O_P\left(N^{-1/2}\right)\right) \\
 & = \bar{M}^{-1}\bar{b} + O_P\left(N^{-1/2}\right) 
 \end{align*}
 where the first equality is by definition of $M$ and $b$; the second equality by McDiarmid; and the third equality follows from uniform continuity as implied by \eqref{eq: my inverse bound}.
 
 To complete the proof, we must show that the first $d$ entries of $\bar{M}^{-1}\bar{b}$ equal zero. To do this, we treat cases \ref{as: phi 1} and \ref{as: phi 2} separately:
 \begin{enumerate}
  \item Under \ref{as: phi 1}, $v_i$ is given by
  \[ v_i = \left[\begin{array}{c} \xi_i(X) \\ \mathbb{E}\xi_i(X) \end{array}\right], \]
  and hence $\bar{M}$ and $\bar{b}$ can be seen to equal
  \[ \bar{M} = \frac{1}{N}\left[\begin{array}{cc} \bar{M}_{1} & \bar{M}_{2} \\ \bar{M}_{2} & \bar{M}_{2}\end{array} \right] \qquad \bar{b} = \frac{1}{N}\left[\begin{array}{c} \bar{b}_1 \\ \bar{b}_1 \end{array}\right], \]
  for $\bar{M}_1$, $\bar{M}_2$ and $\bar{b}_1$ given by 
\begin{align*}
\bar{M}_{1} &=  \sum_{i=1}^N \mathbb{E}[\xi_i(X) \xi_i(X)^T] & 
\bar{M}_2 &= \sum_{i=1}^N \mathbb{E}\xi_i(X) \mathbb{E}\xi_i(X)^T \\ 
\bar{b}_1 & = \sum_{i =1}^N \theta_i \mathbb{E}\xi_i(X)
\end{align*}
 As $\bar{M}$ is invertible, it follows that its submatrix $\bar{M}_2$ is also invertible, and hence that 
  \[ \left[\begin{array}{cc} \bar{M}_{1} & \bar{M}_{2} \\ \bar{M}_{2} & \bar{M}_{2}\end{array} \right]  \left[\begin{array}{c} 0 \\ \bar{M}_2^{-1} b_1 \end{array}\right] = \left[\begin{array}{c} \bar{b}_1 \\ \bar{b}_1 \end{array}\right],\]
  proving that the first $d$ entries of $\bar{M}^{-1}\bar{b}$ equal zero. This proves the claim.

  \item To show the result under \ref{as: phi 2}, let $n_k$ and $m_k$ for $k=1,\ldots,K$ be given by
  \[ n_k = \sum_{i=1}^N 1\{\kappa(i) = k\} \qquad m_k = \sum_{i=1}^N \theta_i 1\{\kappa(i) = k\} \]
  Then $\bar{M}$ and $\bar{b}$ are given by 
\[\bar{M} = \frac{1}{N}\left[\begin{array}{cc} \bar{M}_{11} & \bar{M}_{12} \\ \bar{M}_{12}^T & \bar{M}_{22} \end{array}\right] \qquad \bar{b} = \frac{1}{N}\sum_{k=1}^K m_k \left[\begin{array}{c} \bar{\xi}_k \\ e_k\end{array}\right],\]
for $\bar{M}_{11}$, $\bar{M}_{12}$, and $\bar{M}_{22}$ given by
\begin{align*}
\bar{M}_{11}  & = \sum_{i=1}^N \mathbb{E} \xi_i(X) \xi_i(X)^T &
\bar{M}_{12} & = \left[ n_1 \bar{\xi}_1\ \cdots \ n_K \bar{\xi}_K \right] \\ 
\bar{M}_{22} & = \operatorname{diag}(n_1,\ldots,n_K)
\end{align*}
It follows that
\[ \left[\begin{array}{c} \bar{M}_{12} \\ \bar{M}_{22} \end{array}\right]  =  \left[\begin{array}{ccc} \bar{\xi}_1 & \cdots & \bar{\xi_K} \\ e_1 & \cdots & e_K \end{array}\right] \cdot \operatorname{diag}(n_1,\ldots,n_K)\]
As $\bar{M}$ is invertible and $\bar{b}$ is in the span of the columns of $\left[\begin{array}{c} \bar{M}_{12} \\ \bar{M}_{22} \end{array}\right]$, it can be seen that the first $d$ entries of $\bar{M}^{-1}\bar{b}$ must equal zero, proving the result.
\end{enumerate}
 \end{proof}

\subsection{Proof of Proposition \ref{th: regression normality}}

Our proof will use the following tools. 

Theorem \ref{th: CLT dependency graph} is a central limit theorem  \citepappendix[Thm 2.2]{rinott1994normal} that requires the notion of a dependency graph: Given a collection of dependent random variables $X_1,\ldots,X_N$, a graph with $N$ nodes is said to be its dependency graph if for any two disjoint subsets $S_1$ and $S_2$, if there is no edge in the graph connecting $S_1$ to $S_2$ then the sets of random variables $\{X_i: i \in S_1\}$ and $\{X_i: i \in S_2\}$ are independent.

\begin{theorem}\citeappendix[Thm 2.2]{rinott1994normal} \label{th: CLT dependency graph}
  Let $X_1,\ldots,X_N$ be zero-mean random variables having a dependency graph whose maximal degree is strictly less than $D$, satisfying $|X_i - \mathbb{E}X_i| \leq B$ a.s., $i=1,\ldots,N$, and $\operatorname{Var}\sum_{i=1}^N X_i = \sigma^2 > 0$. Let $\Phi$ denote the CDF of a standard normal. Then for all $t \in \mathbb{R}$ it holds that 
\begin{align*}
 & \left| \mathbb{P}\left(\frac{\sum_{i=1}^N X_i }{\sigma} \leq t\right) - \Phi(t) \right| \\ 
 & \hskip.3cm {} \leq \frac{1}{\sigma} \left\{ \sqrt{\frac{1}{2\pi}} DB + 16 \left(\frac{n}{\sigma^2}\right)^{1/2}\ D^{3/2}B^2 + 10\left(\frac{n}{\sigma^2}\right)D^2B^3 \right\}
\end{align*}
  
\end{theorem}

Lemma \ref{le: dependency graph} bounds the degree of the dependency graph for the thresholded exposures.

\begin{lemma}\label{le: dependency graph}
  Let the assumptions and definitions of Proposition \ref{th: regression normality} hold. Then the random vectors $\{v_i\}_{i=1}^N$ have dependency graph with maximum degree bounded by $d_{\max}^2$.
\end{lemma}

\begin{proof}[Proof of Lemma \ref{le: dependency graph}]
  Let $G'$ denote the graph with $N$ nodes which has an edge between $i$ and $j$ if $G_{i\ell} = G_{j\ell} = 1$ for any node $\ell$, so that $\ell$ is a common neighbor of both $i$ and $j$. Since $G$ has maximum degree $d_{\max}$, it can be seen that $G'$ has maximum degree $d_{\max}^2$.
  
  To see that $G'$ is the dependency graph for $\{v_i\}_{i=1}^N$, we observe that since treatments are assigned independently, if nodes $i$ and $j$ are not connected in $G'$, they have no common neighbors in $G$, and hence $v_i$ and $v_j$ depend on disjoint treatments and are independent. Similarly, given disjoint subsets $S_1$ and $S_2$ with no edges between them in $G'$, it follows that the sets $\{v_i: i \in S_1\}$ and $\{v_i: i \in S_2\}$ also depend on disjoint treatments and are are independent, making $G'$ the dependency graph of $\{v_i\}_{i=1}^N$.
\end{proof}

We will also use the mean-value theorem, which states that if a function $f$ is differentiable on an open set, it holds for any $x$ and $y$ in the set that
\[ f(x) - f(y) = \nabla f(z)^T (x- y),\]
where $z = tx + (1-t)y$ for some $t \in [0,1]$

Finally, we will use the following identity \citeappendix[Eq. 61]{petersen2008matrix}, which takes a matrix $M$ and vectors $a$ and $b$, and returns the derivative of $a^T M^{-1} b$ with respect to $M$:
\begin{equation} \label{eq: matrix inverse derivative}
 \frac{\partial a^T M b}{\partial M} = -M^{-T}ba^TM^{-T} 
\end{equation}

\begin{proof}[Proof of Proposition \ref{th: regression normality}]

The estimation error of $(\hat{\beta}_\ell - \beta_\ell)$ is given by
\begin{align}
\label{eq: regression normality h1} \hat{\beta}_\ell - \beta_\ell &= e_\ell^T \left(\frac{1}{N} \sum_{i=1}^N v_i v_i^T\right)^{-1} \cdot \frac{1}{N} \sum_{i=1}^N \theta_i v_i \\
\label{eq: regression normality h2} & = h\left(\frac{1}{N} \sum_{i=1}^N v_i v_i^T,\,  \frac{1}{N} \sum_{i=1}^N \theta_i v_i \right) \\
 \label{eq: regression normality h3} & = h(T)
\end{align}
where $h$ is the mapping $h(M,b) = e_\ell^T M^{-1}b$ and $T \equiv T(X)$ encodes the arguments of $h$ as a vector. It can be seen that $T$ equals
\begin{align} 
\nonumber T & = \left[\begin{array}{c} \operatorname{vec}\left(\frac{1}{N} \sum_{i=1}^N v_i v_i^T \right) \\ \frac{1}{N}\sum_{i=1}^N \theta_i v_i \end{array}\right] \\
\label{eq: T} & = \frac{1}{N} \sum_{i=1}^N \psi_i,
\end{align}
where $\psi_i \equiv \psi_i(X; \theta) = \left[\begin{array}{c} \operatorname{vec}\left(v_i v_i^T\right) \\ \theta_i v_i \end{array}\right].$ 
The gradient $\nabla h$ is given by
\begin{align} \label{eq: nabla h}
\nabla h(T) = \left[ \begin{array}{c} \operatorname{vec}(-M^{-T}b e_\ell^TM^{-T}) \\ M^{-1}e_\ell \end{array}\right]
\end{align}

The proof begins by decomposing $h(T) - h(\mathbb{E}T)$ using the mean value theorem:
\begin{align}
\nonumber & h(T) - h(\mathbb{E}T) \\
\label{eq: regression normality g1}  & =  \nabla h(\tilde{T})^T (T - \mathbb{E}T) \\
\nonumber & =  \nabla h(\mathbb{E}T)^T(T - \mathbb{E}T) + \left(\nabla h(\tilde{T}) - \nabla h(\mathbb{E}T)\right)^T(T - \mathbb{E}T) \\
\nonumber & = S + E,
\end{align}
where $S$ and $E$ are the terms
\begin{align*}
 S & = \nabla h(\mathbb{E}T)^T(T - \mathbb{E}T) & E & = \left(\nabla h(\tilde{T}) - \nabla h(\mathbb{E}T)\right)^T(T - \mathbb{E}T),
 \end{align*}
 and $\tilde{T}$ is a convex combination of $T$ and $\mathbb{E}T$. The next step is to show the following three intermediate results:
\begin{enumerate}
  \item $S = \ddot{w}^T\theta$ and $h(\mathbb{E}T) = \bar{w}^T\theta$
  \item $S/\sqrt{\operatorname{Var}(S)}$ is asymptotically normal if $\operatorname{Var}(S) \geq N^{-5/4}$.
  \item $E = O_P(N^{-1})$
\end{enumerate}


\paragraph{1. $\mathbf{S = \ddot{w}^T\theta}$ and $\mathbf{h(\mathbb{E}T) = \bar{w}^T\theta}$: } Let $M$ and $b$ denote
\begin{align*}
 M &= \frac{1}{N} \sum_{i=1}^N v_iv_i^T & b & = \frac{1}{N} \sum_{i=1}^N v_i \theta_i,
 \end{align*}
and let $\bar{M} = \mathbb{E}M$ and $\bar{b} = \mathbb{E}b$. To show that $S = \ddot{w}^T\theta$, observe that 
\begin{align*}
S & = \nabla h(\mathbb{E}T)^T (T - \mathbb{E}T) \\
  & = \operatorname{vec}(-\bar{M}^{-1} \bar{b} e_\ell^T \bar{M}^{-1})^T \operatorname{vec}(M - \bar{M}) + e_\ell^T\bar{M}^{-1}(b - \bar{b}) \\
  & = \operatorname{Tr}\left[-\bar{M}^{-1}e_\ell \bar{b}^T \bar{M}^{-1}(M - \bar{M})\right] + e_\ell^T\bar{M}^{-1}(b - \bar{b}) \\
  & = -\bar{b}^T \bar{M}^{-1}(M - \bar{M})\bar{M}^{-1}e_\ell  + e_\ell^T\bar{M}^{-1}(b - \bar{b}) \\
  & = -\bar{b}^T \bar{M}^{-1}M\bar{M}^{-1}e_\ell  + e_\ell^T\bar{M}^{-1}b  + \underbrace{e_\ell^T\bar{M}\bar{b} - e_\ell^T\bar{M}\bar{b}}_{=0}\\
  & = -e_\ell^T \bar{M}^{-1}(M\bar{M}^{-1}\bar{b}  - b)  \\
  & = -\sum_{i=1}^N e_\ell^T \left(\sum_{j=1}^N \mathbb{E}v_iv_i^T\right)^{-1}\left(\sum_{j=1}^N \mathbb{E}v_j v_j^T \left(\sum_{k=1}^N \mathbb{E}v_kv_k^T\right)^{-1}\mathbb{E}v_i  - v_i \right)\theta_i  \\
  & = \sum_{i=1}^N \ddot{w}_i\theta_i
\end{align*}
and that
\begin{align*}
h(\mathbb{E}T) & = e_\ell^T\bar{M}^{-1}\bar{b} \\
& = e_\ell^T \left(\sum_{i=1}^N \mathbb{E}v_iv_i^T\right)^{-1} \sum_{i=1}^N \mathbb{E}v_i \theta_i \\
& = \sum_{i=1}^N \bar{w}_i \theta_i
\end{align*}

\paragraph{2. Normality of $\mathbf{S/\sqrt{\textbf{Var}(S)}}$ if Var$\mathbf{(S) \geq N^{-5/4}}$:} The quantity $S/\sqrt{\operatorname{Var}(S)}$ can be written as
\begin{align*}
\frac{S}{\sqrt{\operatorname{Var}{S}}} & = \frac{1}{N \sqrt{\operatorname{Var}(S)}} \sum_{i=1}^N \nabla h(\mathbb{E}T)^T (\psi_i - \mathbb{E}\psi_i) \\
& = \sum_{i=1}^N s_i
\end{align*}
where $s_i = (N \sqrt{\operatorname{Var}(S)})^{-1} \nabla h(\mathbb{E}T)^T (\psi_i - \mathbb{E}\psi_i)$. It can be seen that the variables $\{s_i\}$ are zero mean, and as each $s_i$ is determined by $\psi_i$ (and hence by $v_i$), they satisfy a bounded dependence condition by Lemma \ref{le: dependency graph}. Finally, to bound $|s_i|$ we bound $\|\nabla h(\mathbb{E}T)\|_\infty$ and $\|\psi_i - \mathbb{E}\psi_i\|_\infty$ separately. To bound $\|\nabla h(\mathbb{E}T)\|_\infty$ we observe that
\begin{align*}
\| \nabla h(\bar{T}) \|_\infty & = \max\left( \| \bar{M}^{-1} \bar{b} e_\ell^T \bar{M}^{-1} \|_\infty, \|\bar{M}^{-1}e_\ell\|_\infty\right) \\
& = \max \left( \| \bar{M}^{-1} \bar{b}\|_\infty \cdot \|\bar{M}^{-1} e_\ell\|_\infty, \|\bar{M}^{-1}e_\ell\|_\infty \right)\\
& \leq \max \left( \| \bar{M}^{-1} \bar{b}\|_2 \cdot \|\bar{M}^{-1} e_\ell\|_2, \|\bar{M}^{-1}e_\ell\|_2 \right)\\
& \leq \max \left( \lambda_{\min}^{-2} B(d+K), \lambda_{\min} \right)
\end{align*}
where the final inequality holds because the smallest eigenvalue of $M$ is lower bounded by $\lambda_{\min}$ by Assumption \ref{as: lambda min} and because $\|\bar{b}\|_\infty \leq B$ by Assumption \ref{as: bounded vi}. To bound $\|\psi_i - \mathbb{E}\psi_i\|_\infty$ we observe that 
\[\|\psi_i - \mathbb{E}\psi_i\|_\infty \leq \max(B, B^2)\]
as $\psi_i$ and $\mathbb{E}\psi_i$ are quadratic functions of $v_i$ and $\mathbb{E}v_i$, which are bounded by Assumptions \ref{as: bounded vi}. As a result, if $\operatorname{Var}(S) \geq N^{-5/4}$, it holds that 
\begin{align*}
|s_i| &= \left|\frac{1}{N \sqrt{\operatorname{Var}(S)}} \nabla h(\mathbb{E}T)^T (\psi_i - \mathbb{E}\psi_i)\right| \\
& \leq N^{-3/8} \, \|\nabla h(\mathbb{E}T)^T \|_\infty\,  \|\psi_i - \mathbb{E}\psi_i\|_\infty \\
& = O\left(N^{-3/8}\right)
\end{align*}
Using these bounds with Theorem \ref{th: CLT dependency graph} implies that
\begin{align*}
\left|\mathbb{P}\left( \frac{S}{\sqrt{\operatorname{Var}(S)}} \leq t\right) - \Phi(t)\right| 
= O\left(N^{-1/8}\right),
\end{align*}
 and hence that $\operatorname{Var}(S) \geq N^{-5/4}$ implies asymptotic normality of $S/\sqrt{\operatorname{Var}}(S)$.

\paragraph{3. $\mathbf{E = O_P(N^{-1})}$:} We first observe that the proof of Proposition \ref{th: regression consistency} implies that $\|T - \mathbb{E}T\| = O_P(N^{-1/2})$. As $\tilde{T}$ is a convex combination of $T$ and $\mathbb{E}T$, it follows that $\|\tilde{T} - \mathbb{E}T\| = O_P(N^{-1/2})$, and hence for  $\nabla h$  given by \eqref{eq: nabla h} that $\| \nabla h(\tilde{T}) - \nabla h(\mathbb{E}T)\| = O_P(N^{-1/2})$. It follows that
\begin{align*}
E &=  \left(\nabla h(\tilde{T}) - \nabla h(\mathbb{E}T)\right)^T (T - \mathbb{E}T) \\
& \leq \|\nabla h(\tilde{T}) - \nabla h(\mathbb{E}T)\| \cdot \| (T - \mathbb{E}T)\| \\
& = O_P(N^{-1/2}) \cdot O_P(N^{-1/2}) \\
& = O_P(N^{-1})
\end{align*}

\paragraph{Putting it together:} Divide the sequence of experiments into two subsequences according to whether or not $\operatorname{Var}(S) \geq N^{-5/4}$. For the subsequence of experiments where $\operatorname{Var}(S) \geq N^{-5/4}$ holds, it follows that
\begin{align*}
\frac{h(T) - \mathbb{E}T}{\sqrt{\operatorname{Var}(S)}} & = \frac{S}{\sqrt{\operatorname{Var}(S)}} + \frac{E}{\sqrt{\operatorname{Var}(S)}} \\
& = \frac{S}{\sqrt{\operatorname{Var}(S)}} + o_P(1) 
\end{align*}
and hence that 
\begin{equation}
\label{eq: normality result} \frac{h(T) - \mathbb{E}T)}{\sqrt{\operatorname{Var}(S)}} \rightarrow N(0,1)
\end{equation}
when $\operatorname{Var}(S) \geq N^{-5/4}$. Therefore, for this subsequence it holds that
\begin{align}
\nonumber & \mathbb{P}\left( \hat{\beta}_\ell - \beta_\ell \geq  U \right) \\
\label{eq: general result end 0}   & \hskip.5cm = \mathbb{P}\left( h(T) \geq  \max_\vartheta \left[ \bar{w}^T \vartheta + z_{1-\alpha} \sqrt{\operatorname{Var}(\ddot{w}^T\vartheta)} \right]\right) \\
\nonumber  & \hskip.5cm \leq \mathbb{P}\left( h(T) \geq  \bar{w}^T \theta + z_{1-\alpha} \sqrt{\operatorname{Var}(\ddot{w}^T\theta)} \right) \\
\nonumber  & \hskip.5cm = \mathbb{P}\left( \frac{ h(T)  - \bar{w}^T \theta}{\sqrt{\operatorname{Var}(\ddot{w}^T\theta)}} \geq z_{1-\alpha} \right) \\
\label{eq: general result end 1}   & \hskip.5cm  = \mathbb{P}\left( \frac{h(T)  - h(\mathbb{E}T)}{\sqrt{\operatorname{Var}(S)}} \geq z_{1-\alpha} \right) \\
\label{eq: general result end 2}   & \hskip.5cm  =  \alpha + o(1),
\end{align}
where \eqref{eq: general result end 0} holds by definition of $U$ in \eqref{eq: regression normality U}, \eqref{eq: general result end 1} holds because $S = \ddot{w}^T\theta$ and $h(\mathbb{E}T) = \bar{w}^T\theta$, and \eqref{eq: general result end 2} holds by\eqref{eq: normality result}

On the other hand, for the subsequence in which $\operatorname{Var}(S) < N^{-5/4}$, it follows by Chebychev's inequality that $S = O_P(N^{-5/8})$, and hence that
\begin{align}
 \nonumber h(T) - \mathbb{E}T &= S + E  \\
 \nonumber &= O_P(N^{-5/8}) + O_P(N^{-1})  \\
 \label{eq: op result} & = o_P(N^{-5/9})
\end{align}
and hence that $h(T) - h(\mathbb{E}T) = o_P(N^{-5/9})$. This implies 
\begin{align}
\nonumber & \mathbb{P}\left( \hat{\beta}_\ell - \beta_\ell \geq U' \right) \\
\nonumber & \hskip.5cm = \mathbb{P}\left( h(T) \geq \max_{\vartheta} \left[ \bar{w}^T \vartheta + \epsilon N^{-5/9} \right]\right) \\
\nonumber & \hskip.5cm \leq \mathbb{P}\left( h(T) \geq \bar{w}^T\theta + \epsilon N^{-5/9}\right) \\
\nonumber & \hskip.5cm \leq \mathbb{P}\left( \frac{ h(T) - \bar{w}^T\theta}{N^{-5/9}} \geq \epsilon\right) \\
\nonumber & \hskip.5cm = \mathbb{P}\left( \frac{ h(T) - h(\mathbb{E}T)}{N^{-5/9}} \geq \epsilon\right) \\
\label{eq: general result end 3} & \hskip.5cm = o(1), 
\end{align}
where \eqref{eq: general result end 3} is implied by \eqref{eq: op result}. Combining \eqref{eq: general result end 2} and \eqref{eq: general result end 3} implies that $\max(U, U')$ upper bounds $\hat{\beta}_\ell - \beta_\ell$ with at least the desired asymptotic coverage probability for both subsequences. As $\hat{\beta}_\ell = w^TY$, it follows that $w^TY - \max(U, U')$ lower bounds $\beta_\ell$ with the same probability.

The proof of coverage for the upper bound $w^TY + \min(L, L')$ is analogous.

\end{proof}

\section{Summary Information for Data Sets} \label{sec: data tables}

Here we give tables to summarize the simulated cholera trial outcomes released by \cite{perez2014assessing} (and available in the  R package \texttt{inferference} described in \citeappendix{saul2017recipe}), insurance takeup by farmers in the experiment of \cite{cai2015social}, and wristband-wearing decisions made by students in the school interventions described in \citeappendix{paluck2016changing}. An aggregate table of the actual outcomes for the cholera trial is given by Table \ref{table: vaccine real data}.

\begin{table}[h!]
\begin{center}
  \begin{tabular}{l|r|r}
     Vaccination Rate & Cases, placebo & Cases, vaccine   \\
     \hline
     0-25\% &   50/111 (45\%) &  37/144 (26\%)\\
     25-50\% & 73/282 (25\%)  &  81/464 (17\%) \\
     50-75\% &  24/183 (13\%)& 42/480 (9\%) \\
     75-100\% &  2/13 (15\%) & 4/110 (4\%) \\
     \hline
     Total & 149/589 (25\%) & 164/1198 (14\%)
  \end{tabular}
\end{center}
\caption{Outcomes for simulated cholera trial dataset}
\label{table: cholera sim summary}
\end{table}

A comparison of Tables \ref{table: cholera sim summary} and \ref{table: vaccine real data} shows that the actual population is much larger than the simulated data set ($N$=74K vs 1.7K), and has much lower rates of cholera incidence (0.003 vs 0.17). As the simulated outcomes were fit to the actual data as a function of treatment (including neighborhood vaccination) and observed covariates (age and distance to water), we speculate that the covariate distributions were shifted to a higher-risk subpopulation.

\cite{perez2014assessing} gives some information regarding the simulation parameters, stating that the outcome for unit $j$ in neighborhood $i$ was determined by a logistic regression model
\[ \text{log odds} P(Y_{ij} = 1) = 0.5 -0.788a - 2.953\alpha_i -0.098 X_{ij1} - 0.145X_{ij2} + 0.35a \alpha_i,\]
where $a$ is the unit treatment; $\alpha_i$ was the neighborhood vaccination rate if $j$'s treatment equals $a$; $X_{ij1}$ and $X_{ij2}$ denote $j$'s age (in decades) and distance to the nearest river (in km), which were randomly generated. Treatments were assigned i.i.d. Bernoulli(2/3). Units were generated by simulating the covariate $X$, and then selecting units to participate in the trial by a logistic regression model
\[ \text{log odds} P(\text{participate}_{ij}) = 0.2727 - 0.0387 X_{ij1} + 0.2179 X_{ij2} + b_i,\]
where $b_i$ is a random neighborhood-level effect.

Table \ref{table: cai summary} gives summary statistics for the experiment of \cite{cai2015social}. In this experiment, units were randomized not only by first or second round and by high or low information, but also whether or not the second round units were shown a list of names corresponding to people who had purchased the insurance product in the first round. When studying the effects of informal communication between rounds, the analysis of \cite{cai2015social} is restricted to the ``no info'' group which excludes those units who were shown such a list, and we follow their example in our analysis as well.

\begin{table}[h!]
\begin{center}
  \begin{tabular}{l|r}
Exposure  & Insurance Takeup \\
(\# of high info friends) &  \\
     \hline
     0 &  116/293 (40\%) \\
     1 &  196/454 (43\%) \\
     2 &  96/171 (56\%)\\
     3 &  21/46 (46\%)\\
     4 & 5/8 (63\%) \\
     5 &  1/1 (100\%)\\
     \hline    
     Total & 435/973 (45\%)
  \end{tabular}
\end{center}
\caption{Outcomes for insurance purchases of farmers, for ``no info'' group used by \cite{cai2015social}}
\label{table: cai summary}
\end{table}

Table \ref{table: paluck summary} gives summary statistics for the experiment of \citeappendix{paluck2016changing}. For this experiment, the analysis of the wristband outcomes in \citeappendix{paluck2016changing} and \cite{aronow2017estimating} restricts to units who had positive probability of being assigned to one of the four treatment outcomes that were considered, meaning that they were eligible units with at least one eligible friend, and we follow their example in our analysis as well.

\begin{table}[h!]
\begin{center}
  \begin{tabular}{l|l|r}
Treated? & $\geq 1$ Treated Friend?  & Outcome \\
     \hline
No & No & 12/158 (7\%) \\
No & Yes & 61/374 (16\%) \\
Yes & No & 53/201 (26\%)\\
Yes & Yes & 103/333 (31\%)\\
\hline
\multicolumn{2}{r}{Total} & 229/1066 (21\%)
  \end{tabular}
\end{center}
\caption{Outcomes for school intervention}
\label{table: paluck summary}
\end{table}

\vskip3cm

 \bibliographystyleappendix{apalike}
 {\footnotesize
 \bibliographyappendix{bibfile}
}

\end{document}